\documentclass[thesis,11pt,oneside]{amsart}%
\usepackage{amssymb}
\usepackage{amsmath}
\usepackage{mathrsfs}
\usepackage{amsfonts}
\usepackage{graphicx}
\usepackage[left=1.5in, right=1.5in, top=1in, bottom=1in, includefoot, headheight=13.6pt]{geometry}
\usepackage[all]{xy}
\newtheorem*{maintheorem}{Main Theorem}
\theoremstyle{plain}

\newtheorem{corollary}{Corollary}

\newtheorem*{definition}{Definition}

\newtheorem*{lemma}{Lemma}

\newtheorem{proposition}{Proposition}
\newtheorem*{remark}{Remark}

\numberwithin{equation}{section}
\begin{document}
\title[Q-adapted It\^{o} formula of Stochastic Calculus]{Multiple $\mathrm{Q}
$-Adapted Integrals and It\^{o} Formula of Noncommutative Stochastic
Calculus in Fock Space}
\author{V. P. Belavkin}
\address{Mathematics Department, University of Nottingham,\\
NG7 2RD, UK.}
\email{vpb@maths.nott.ac.uk}
\thanks{}
\author{M. F. Brown.}
\address{Mathematics Department, University of Nottingham,\\
NG7 2RD, UK.}
\email{pmxmb1@nottingham.ac.uk}
\date{May 2011}
%\subjclass[2000]{Primary 60H99; Secondary 60G99.}
\keywords{Quantum dynamics, Noncommutative analysis, $\mathrm{Q}$-adapted
process.}

\begin{abstract}
We study the continuity property of multiple Q-adapted quantum stochastic
integrals with respect to noncommuting integrands given by the non-adapted
multiple integral\ kernels\ in Fock scale. The noncommutative algebra of
relatively (exponentially) bounded nonadapted quantum stochastic processes
is studied in the kernel form as introduced in \cite{Be91}. The differential
Q-adapted formula generalizing It\^{o} product formula for adapted integrals
is presented in both strong and weak sense as a particular case of the
quantum stochastic nonadapted It\^{o} formula.
\end{abstract}

\maketitle
\tableofcontents
\section{Introduction} Non-commutative
generalization of adapted It\^{o} stochastic calculus, developed by Hudson
and Parthasarathy (HP) in \cite{HudP84}, gave an adequate mathematical tool
for studying unitary and endomorphic cocycles of open quantum dynamical
systems singularly interacting with a boson quantum-stochastic field. The
adapted HP quantum stochastic calculus and its kernel variant {\cite{Mey87}}%
, {\cite{LinM88a}} also made it possible to solve the old quantum
measurement problem by describing such systems by quantum stochastic
Langevin equation {\cite{EvaH88}} with continuously observed output field
\cite{Be89d}, and constructing a quantum filtration theory {\cite{Be92}}
which explained a continuous spontaneous collapse under such observation in
terms of now famous quantum stochastic Master equations, first derived in
\cite{Be88a},\cite{Be91c}. The Belavkin filtering equations gave examples of
quantum stochastic non-unitary, even nonlinear, non-stationary, adapted
evolution equations in Hilbert and operator spaces normalized only in the
mean as the exponential martingales. Their solutions require a proper
definition of chronologically ordered unbounded quantum stochastic
exponentials of \ noncommuting operators and maps which cannot be studied
within the original HP-calculus approach and its extensions {\cite{AccF88},%
\cite{AccQ89}},{\cite{ParS86}} to the bounded QS (quantum stochastic)
semi-martingales. Moreover, the perturbation theory for such evolutions,
usually studied by applying Duhamel principle in non-stochastic case,
requires the development of quantum non-adapted calculus, since the
stochastic Duhamel formula cannot be written in terms of the adapted
stochastic integral even in the classical case. In order to solve these
problems, a more general Fock scale approach to quantum stochastic calculus
and integration was developed in \cite{Be88},{\cite{Be91},\cite{Be92c}}
which does not require the usual boundedness and adaptedness of the
stochastic integrands and the resulting quantum integral semimartingales.
Based on the canonical Pseudo-Poisson representation of quantum It\^{o}
algebra discovered by Belavkin in \cite{Be88}, the Fock scale analysis gives
a simple constructive expression of the quantum functional It\^{o} formula
and is essentially basis free, allowing a unified algebraic treatment of
any, even infinite number of quantum noise modes.

In this paper we give a brief overview of nonadapted quantum stochastic
calculus closely following the nonadapted quantum It\^{o} formula part of
the paper {\cite{Be92c}, but consider the case of }$\mathrm{Q}$-adapted
quantum stochastic integrals introduced in \cite{BelB10} in the natural Fock
scale of Hilbert spaces as a special nonadapted case of {\cite{Be92c}}. To
this end we shall explore the Belavkin notation for indefinite $\star $%
-algebraic structure of the kernel calculus as the general property of a
natural pseudo-Euclidean representation for Sch\"{u}rmann's tripples
associated with infinitely divisible states, obtained by Belavkin for the
general nonstationary case in {\cite{Be92c}}. As a particular case of
non-adapted It\^{o} formula we establish a non-commutative $\mathrm{Q}$%
-adapted generalization of the adapted It\^{o} formula that is the principal
formula of the classical stochastic calculus. In the $\mathrm{Q}=\mathrm{I}$
case this formula coincides with the well-known Hudson-Parthasarathy formula
\cite{HudP84} for the product of a pair of non-commuting quantum processes
and gives its functional extension. In the commutative case this gives a $%
\mathrm{Q}$-adapted generalization of the It\^{o} formula for classical
stochastic processes as the case of the general nonadapted classical It\^{o}
formula, first obtained in the case of Wiener integrals in a weak form by
classical stochastic methods by Nualart \cite{NuaP88}. We also note that
Fock scales are also used in the stochastic analysis of classical white
noise by Hida, Kuo, Pothoff, Streit, starting from \cite{Hid80},\cite{PotS89}%
, and independently by Berezanskii and Kondrat'ev, starting from \cite%
{BerK88}. However, while the classical stochastic analysis is mostly
concerned with the study of the generalized stochastic functionals in
nuclear Fock scales, the quantum noise analysis is concerned with the
analysis of generalized operators (kernels) in nonnuclear Fock scale which
was first introduced by Belavkin \cite{Be88},{\cite{Be91},\cite{Be92c}} and
recently used also by Ji and Obata in their abstract Fock space approach to
nondifferential quantum stochastic analysis.

Here we shall consider mostly differetial problems of quantum stochastic
analysis, adopting Guichardet representation of Fock space as $\mathit{L}^{2}
$-space over the finite subsets from a nonatomic measure space, regarding
these subsets as almost totally ordered chains following the notation from {%
\cite{Be91},\cite{Be92c}}. In this notation the integral quantum stochastic
calculus is similar in spirit to the kernel calculus of
Maassen-Lindsay-Meyer {\cite{LinM88a}}, {\cite{Mey87}}, with the difference
that all the main objects are constructed not in terms of Maassen-Meyer
kernels but in terms of the operators kernels represented in the Fock state
space. In this unifying approach we employ a much more general notion of
multiple stochastic integral, non-adapted in general, but focusing on $%
\mathrm{Q}$-adapted processes, which reduces to the notion of the kernel
representation of an operator only in the case of a scalar (non-random)
operator integrand. The possibility of defining a non-adapted single
integral in terms of the kernel calculus in the case of quantum single mode
noise was shown independently by Belavkin \cite{Be88} and Lindsay {\cite%
{Lin90}}, but the general repeated and the multiple nonadapted integrals
were first introduced and studied within the quantum stochastic analysis in
Fock scales in {\cite{Be91},\cite{Be92c}}.

\section{Rigged Guichardet-Fock Space}

Let $(\mathbb{X},\lambda)$ be an essentially ordered space, that is, a
measurable space $\mathbb{X}$ with a $\sigma$-finite measure $\lambda:%
\mathfrak{F}_{\mathbb{X}}\ni \boldsymbol{\bigtriangleup}\mapsto\lambda\left(
\boldsymbol{\bigtriangleup }\right) \geq0$ and an ordering relation $x\leq
x^{\prime}$ with the property that any $n$-tuple $\left(
x_{1},\ldots,x_{n}\right) \in \mathbb{X}^{n}$ can be identified up to a
permutation with a chain $\vartheta=\{x_{1}<\cdots<x_{n}\}$ modulo the
product measure $\prod_{i=1}^{n}\mathrm{d}x_{i}$ of $\mathrm{d}%
x:=\lambda\left( \mathrm{d}x\right) $. In particular we consider the order
induced by the linear order in $\mathbb{R}_+$ by a measurable map $t:\mathbb{%
X}\rightarrow\mathbb{R}_+$ relatively to which $\lambda$ is absolutely
continuous with respect to the Lebesque measure $\mathrm{d}t$ on $\mathbb{R}%
_+$, then $x<x^\prime\Leftrightarrow t(x)<t(x^\prime)$, where $t(x)$ is
identified with the time of the point $x$.

We shall identify the finite chains $\vartheta$ with increasingly indexed $n
$-tuples $\boldsymbol{s}\equiv(x_{1},\ldots,x_{n})$ with $x_{i}\in \mathbb{X}
$, $x_{1}<\cdots<x_{n}$, denoting by $\mathcal{X}=\sum_{n=0}^{\infty}%
\mathcal{X}_{n}$ the set of all finite chains as the union of the sets
\begin{equation*}
\mathcal{X}_{n}=\{\boldsymbol{s}\in \mathbb{X}^{n}:x_{1}<\cdots<x_{n}\}
\end{equation*}
with $\mathcal{X}_{0}=\{\emptyset \}$ containing the only one element $%
\emptyset \in \mathcal{X}_{0}$, the empty chain $\emptyset =\mathbb{X}^{0}$
identified with empty subset of $\mathbb{X}$. We introduce a measure
`element' $\mathrm{d}\vartheta=\prod_{x\in\vartheta}\mathrm{d}x$ on $%
\mathcal{X}$ induced by the direct sum $\oplus_{n=0}^{\infty}\lambda^{%
\otimes n}\left( \boldsymbol{\bigtriangleup}_{n}\right) ,\boldsymbol{%
\bigtriangleup}_{n}\in\mathfrak{F}_{\mathbb{X}}^{\otimes n}$ of product
measures $\mathrm{d}\boldsymbol{s}=\prod_{i=1}^{n}\mathrm{d}x_{i}$ on $%
\mathbb{X}^{n}$ with the unit mass $\mathrm{d}\vartheta=1$ at the only
atomic point $\vartheta=\emptyset$.

Let $\{\mathfrak{k}_{x}:x\in \mathbb{X}\}$ be a family of Hilbert spaces $%
\mathfrak{k}_{x}$, let $\mathfrak{p}_{0}$ be an additive semigroup of
nonnegative essentially measurable locally bounded functions $q:\mathbb{X}%
\rightarrow\mathbb{R}_{+}$ with zero included $0\in\mathfrak{p}_{0}$, and
let $\mathfrak{p}_{1}=\{1+q_{0}:q_{0}\in\mathfrak{p}_{0}\}$. We denote by $%
\mathit{K}_{\star}(q)$ the Hilbert space of essentially measurable
vector-functions $\mathrm{k}:x\mapsto\mathrm{k}(x)\in\mathfrak{k}_{x}$ which
are square integrable with the weight $q\in\mathfrak{p}_{1}$:
\begin{equation*}
\left\Vert \mathrm{k}\right\Vert (q)=\left( \int\left\Vert \mathrm{k}%
(x)\right\Vert _{x}^{2}q(x)\mathrm{d}x\right) ^{1/2}<\infty.
\end{equation*}
With $q\geq1$, any space $\mathit{K}_{\star}(q)$ can be embedded into the
Hilbert space $\mathcal{K}_\ast=\mathit{K}_{\star}(1)$, and the intersection $%
\cap_{q\in\mathfrak{p}_{1}}\mathit{K}_{\star}(q)\subseteq\mathfrak{k}$ can
be identified with the projective limit $\mathit{K}_{+}=\lim
_{q\rightarrow\infty}\mathit{K}_{\star}(q)$. This follows from the facts
that the function $\left\Vert \mathrm{k}\right\Vert (q)$ is increasing: $%
q\leq p\Rightarrow\left\Vert \mathrm{k}\right\Vert (q)\leq\left\Vert \mathrm{%
k}\right\Vert (p)$, and so $\mathit{K}_{\star}(p)\subseteq\mathit{K}%
_{\star}(q)$, and that the set $\mathfrak{p}_{1}$ is directed in the sense
that for any $q=1+r$ and $p=1+s$, $r,s\in\mathfrak{p}_{0}$, there is a
function in $\mathfrak{p}_{1}$ majorizing $q$ and $p$ (we can take for
example $q+p-1=1+r+s\in \mathfrak{p}_{1}$).

The dual space $\mathit{K}_{\star}^{-}$ to$\ \mathit{K}_{+}$ is the space of
generalized vector-functions $\mathrm{f}\left( x\right) $ defining the
continuous functionals
\begin{equation*}
\left\langle \mathrm{f}\mid\mathrm{k}\right\rangle =\int\left\langle \mathrm{%
f}(x)\mid\mathrm{k}(x)\right\rangle \,\mathrm{d}x,\quad\mathrm{k}\in\mathit{K%
}_{+}.
\end{equation*}
It is the inductive limit $\mathit{K}_{-}=\lim_{q\rightarrow0}\mathit{K}%
_{\star}(q)$ in the opposite scale $\{\mathit{K}_{\star}(q):q\in\mathfrak{p}%
_{-}\}$, where $\mathfrak{p}_{-}$ is the set of functions $q:\mathbb{X}%
\rightarrow (0,1]$ such that $1/q\in\mathfrak{p}_{1}$, which is the union $%
\cup _{q\in\mathfrak{p}_{-}}\mathit{K}_{\star}(q)$ of the inductive family
of Hilbert spaces $\mathit{K}_{\star}(q),q\in\mathfrak{p}_{-}$, with the
norms $\left\Vert \mathrm{k}\right\Vert (q)$, containing as the minimal the
space $\mathcal{K}_{\ast}$. Thus we obtain the
Gel'fand chain
\begin{equation*}
\mathit{K}_{+}\subseteq\mathit{K}_{\star}(q_{+})\subseteq\mathcal{K}_{\ast
}\subseteq\mathit{K}_{\star}(q_{-})\subseteq\mathit{K}_{-}
\end{equation*}
in the extended scale $\{\mathit{K}_{\star}(q):q\in\mathfrak{p}\}$, where $%
\mathfrak{p}=\mathfrak{p}_{-}\cup\mathfrak{p}_{1}$, with $q_{+}\in\mathfrak{p%
}_{1},\,q_{-}\in\mathfrak{p}_{-}$. The dual space $\mathit{K}_{+}^{\star}=%
\mathit{K}^{-}$ is the space of the continuous linear functionals on $%
\mathit{K}_{+}$ containing the Hilbert space $\mathcal{K}$ called the\emph{\
rigged space} with respect to the dense subspace $\mathit{K}^{+}=\mathit{K}%
_{-}^{\star}$ of $\mathcal{K}$ equipped with the projective convergence in
the scale $\left\Vert \mathrm{k}^{\ast}\right\Vert \left( q\right)
=\left\Vert \mathrm{k}\right\Vert \left( q\right) $ for $q\in\mathfrak{p}_{1}
$.

We can similarly define a Fock-Gel'fand triple $(\mathit{F}_{+},\mathcal{F}%
_{\ast},\mathit{F}_{-})$ with%
\begin{equation*}
\mathit{F}_{+}=\cap_{q\in\mathfrak{p}_{1}}\mathit{F}_{\star}(q),\;\mathcal{F}%
_{\ast}=\mathit{F}_{\star}(1),\;\mathit{F}_{-}=\cup_{q\in\mathfrak{p}_{-}}%
\mathit{F}_{\star}(q),
\end{equation*}
for the Hilbert scale $\{\mathit{F}_{\star}(q):q\in\mathfrak{p}\}$ of the
symmetric Fock spaces $\mathit{F}_{\star}(q)=\oplus_{n=0}^\infty
K_\star^{(n)}(q)$ over $\mathit{K}_{\star}(q)$, where $K_\star^{(0)}(q)=%
\mathbb{C}$, $K_\star^{(1)}(q)=K_\star(q)$, and each $K_\star^{(n)}(q)$ for $%
n>1$ is given by the product-weight $q_n(x_1,\ldots,x_n)=\prod_{i=1}^nq(x_i)$
on $\mathbb{X}^n$. We shall consider the Guichardet {\cite{Gui72}}
representation of the symmetric tensor-functions $\psi_n\in K^{(n)}_\star(q)$
regarding them as the restrictions $\psi|\mathcal{X}_n$ of the functions $%
\psi:\vartheta\mapsto\psi(\vartheta)\in K_\star^\otimes(\vartheta)$ with
sections in the Hilbert products $\mathit{K}_{\star}^{\otimes}(\vartheta)=%
\otimes _{x\in\vartheta}\mathfrak{k}_{x}$, square integrable with the
product weight $q(\vartheta)=\prod_{x\in\vartheta}q(x)$:
\begin{equation*}
\Vert\psi\Vert(q)=\left( \int\Vert\psi(\vartheta)\Vert^{2}q(\vartheta )\,%
\mathrm{d}\vartheta\right) ^{1/2}<\infty.
\end{equation*}
The integral here is over all chains $\vartheta\in\mathcal{X}$ and defines
the pairing on $\mathit{F}_{+}$ by
\begin{equation*}
\left\langle \psi\mid\psi\right\rangle =\int\left\langle \psi(\vartheta
)\mid\psi(\vartheta)\right\rangle \,\mathrm{d}\vartheta,\quad\psi\in \mathit{%
F}_{+}.
\end{equation*}
In more detail we can write this in the form
\begin{equation*}
\int\Vert\psi(\vartheta)\Vert^{2}q(\vartheta)\mathrm{d}\vartheta=\sum
_{n=0}^{\infty}\int\limits_{0\leq t_{1}<}\cdots\int\limits_{<t_{n}<\infty
}\Vert\psi(x_{1},\ldots,x_{n})\Vert^{2}\prod_{i=1}^{n}q(x_{i})\mathrm{d}%
x_{i},
\end{equation*}
where the $n$-fold integrals for $\psi_n\in K_\star^{(n)}$ are taken over
simplex domains $\mathcal{X}_{n}=\{\boldsymbol{s}\in \mathbb{X}%
^{n}:t(x_{1})<\cdots<t(x_{n})\}$.

One can easily establish an isomorphism between the space $\mathit{F}%
_{\star}(q)$ and the symmetric (or antisymmetric) Fock space over $\mathit{K}%
_{\star}(q)$ with a nonatomic measure $\mathrm{d}x$ in $\mathbb{X}$. It is
defined by the isometry
\begin{equation*}
\Vert\psi\Vert(q)=\left( \sum_{n=0}^{\infty}\frac{1}{n!}\idotsint\Vert
\psi(x_{1},\ldots,x_{n})\Vert^{2}\prod_{i=1}^{n}q(x_{i})\mathrm{d}%
x_{i}\right) ^{1/2},
\end{equation*}
where the functions $\psi(x_{1},\ldots x_{n})$ can be extended to the whole
of $\mathbb{X}^{n}$ in a symmetric (or antisymmetric) way uniquely up to the
measure zero due to nonatomicity of $\mathrm{d}x$ on $\mathbb{X}$.

Finally let $\mathfrak{h}$ be a Hilbert space called the initial space for
the Hilbert products $\mathcal{H}_{\ast}=\mathfrak{h}\otimes\mathcal{K}%
_{\ast} $ and $\mathcal{G}_{\ast}=\mathfrak{h}\otimes\mathcal{F}_{\ast}$. We
consider the Hilbert scale $\mathit{G}_{\star}(q)=\mathfrak{h}\otimes\mathit{%
F}_{\star}(q)$, $q\in\mathfrak{p}$ of complete tensor products of $\mathfrak{%
h}$ and the Fock spaces over $\mathit{K}_{\star}(q)$, and we put%
\begin{equation*}
\mathit{G}_{+}=\cap\mathit{G}_{\star}(q),\;\;\mathit{G}_{-}=\cup\mathit{G}%
_{\star}(q)
\end{equation*}
which constitute the Gel'fand triple $\mathit{G}_{+}\subseteq\mathcal{G}%
_{\ast}\subseteq\mathit{G}_{-}$ dual to $\mathit{G}^{+}\subseteq \mathcal{G}%
\subseteq\mathit{G}^{-}$ of the Hermitian adjoint bra-spaces $\mathit{G}^{+}=%
\mathit{G}_{+}^{\ast}$, $\mathcal{G}=\mathcal{G}_{\ast }^{\ast}$, $\mathit{G}%
^{-}=\mathit{G}_{-}^{\ast}$.

\section{Triangular Kernels as Generalized Operators in Fock Space}

Now we consider matrix chains $\vartheta^\cdot_\cdot=[\vartheta^\mu_\nu]^{%
\mu=-,\circ,+}_{\nu=-,\circ,+}$ and we define the triangular operator valued
kernels
\begin{align}
& T(\vartheta^\cdot_\cdot)= 0\text{ if }\vartheta^\mu_\nu\neq0\text{ for
each }\mu>\nu,  \notag \\
& T(\vartheta^\cdot_\cdot)= 1\left( \vartheta_{-}^{-}\right) T\left(
\boldsymbol{\vartheta}\right) 1\left( \vartheta_{+}^{+}\right) \text{
otherwise},
\end{align}
where $1\left( \vartheta_{-}^{-}\right)=1=1\left( \vartheta_{+}^{+}\right)$
in $\mathbb{C}$, and
\begin{equation}
T(\boldsymbol{\vartheta})=T%
\begin{pmatrix}
\vartheta_{+}^{-} & \vartheta_{\circ}^{-} \\
\vartheta_{+}^{\circ} & \vartheta_{\circ}^{\circ}%
\end{pmatrix}
:\mathfrak{k}^{\otimes}(\vartheta_{\circ}^{-}\sqcup\vartheta_{\circ}^{\circ
})\otimes\mathfrak{h}\rightarrow\mathfrak{k}^{\otimes}(\vartheta_{\circ
}^{\circ}\sqcup\vartheta_{+}^{\circ})\otimes\mathfrak{h},   \label{two1}
\end{equation}
is an operator-valued function of $\boldsymbol{\vartheta}=\left(
\vartheta_{\nu}^{\mu}\right)^{\mu=-,\circ}_{\nu=+,\circ} $ satisfying the
integrability condition $\Vert T\Vert_{q}(r)<\infty$ for some $r^{-1}\in%
\mathfrak{p}_{0}$ and $q\in \mathfrak{p}_{1}$ with respect to the norms
\begin{equation*}
\Vert T\Vert_{q}(r)=\int\mathrm{d}\vartheta_{+}^{-}\left( \iint \mathrm{ess}%
\sup_{\vartheta_{\circ}^{\circ}}\left\{\frac{\Vert T(\boldsymbol{\vartheta}%
)\Vert}{q(\vartheta_{\circ}^{\circ})}\right\}^{2}r(\vartheta_{+}^{\circ}%
\sqcup\vartheta_{\circ}^{-})\mathrm{d}\vartheta _{+}^{\circ}\mathrm{d}%
\vartheta_{\circ}^{-}\right) ^{1/2}.
\end{equation*}
Note that $T\left( \boldsymbol{\vartheta}\right) \in\mathfrak{L}\left(
\mathfrak{h}\right) $ in the scalar case $\mathfrak{k}_x=\mathbb{C}$.

We would like to consider the QS integral operators as the continuous maps $%
\mathrm{T}:G_+\rightarrow G_-$ representing the triangular kernels $%
T(\vartheta^\cdot_\cdot)$. The representation, denoted by $\boldsymbol{%
\epsilon}$, is explicitly defined by
\begin{equation}
\lbrack\boldsymbol{\epsilon}(T)\chi](\vartheta)=\sum_{\vartheta_{\circ}^{%
\circ}\sqcup\vartheta_{+}^{\circ}=\vartheta}\iint T%
\begin{pmatrix}
\vartheta_{+}^{-}, & \vartheta_{\circ}^{-} \\
\vartheta_{+}^{\circ}, & \vartheta_{\circ}^{\circ}%
\end{pmatrix}
\chi(\vartheta_{\circ}^{\circ}\sqcup\vartheta_{\circ}^{-})\mathrm{d}%
\vartheta_{\circ}^{-}\mathrm{d}\vartheta_{+}^{-}   \label{two2}
\end{equation}
on $\chi\ \in\mathit{G}_{+}$, which may be given by the operator-valued
multiple integral \cite{Be91}
\begin{equation}
\lbrack\boldsymbol{\imath}_{0}^{t}(\mathrm{M})\chi](\vartheta)=\sum
_{\upsilon_{\circ}^{\circ}\sqcup\upsilon_{+}^{\circ}\subseteq\vartheta^{t}}%
\int_{\mathcal{X}^{t}}\int_{\mathcal{X}^{t}}[\mathrm{M}(\boldsymbol{\upsilon
})\mathring{\chi}(\upsilon_{\circ}^{-}\sqcup\upsilon_{\circ}^{\circ
})](\vartheta_{-}^{\circ})\mathrm{d}\upsilon_{+}^{-}\mathrm{d}%
\upsilon_{\circ }^{-},   \label{2onee}
\end{equation}
on $\mathcal{X}^{\infty }=\mathcal{X}$ of the multiple integrand $\mathrm{M}(%
\boldsymbol{\upsilon })$ as trivially vacuum adapted such that $\mathrm{M}(%
\boldsymbol{\upsilon })=T(\boldsymbol{\upsilon })\otimes \mathrm{P}%
_{\emptyset }$, where $[\mathrm{P}_{\emptyset }\chi ](\vartheta )=\delta
_{\emptyset }(\vartheta )\chi (\vartheta )$, and $\delta _{\emptyset
}(\vartheta )=1$ if $\vartheta =\emptyset $ and is $0$ otherwise, such that $%
[\mathrm{M}(\boldsymbol{\upsilon })\mathring{\chi}(\upsilon _{\circ }^{\circ
}\sqcup \upsilon _{\circ }^{-})](\varkappa )=0$ if $\varkappa \neq \emptyset
$.

The integrand $\mathrm{M}(\boldsymbol{\upsilon })$ is not, however, a unique
choice in (\ref{2onee}), one can take any kernel-valued function $M\left(
\boldsymbol{\upsilon },\boldsymbol{\varkappa }\right) $ satisfying $T\left(
\boldsymbol{\vartheta }\right) =\sum_{\varkappa \sqcup \upsilon _{\circ
}^{\circ }=\vartheta _{\circ }^{\circ }}M\left( \boldsymbol{\upsilon }%
,\varkappa \right) $ to obtain the operator $\boldsymbol{\epsilon }(T)$. We
shall choose $M(\boldsymbol{\upsilon },\varkappa )=M(\boldsymbol{\upsilon }%
)\otimes \mathrm{Q}^{\otimes }\left( \varkappa \right) $ as trivially $%
\mathrm{Q}$-adapted multiple integrand \cite{BelB10} such that the
operator-valued $\mathrm{I}$-adapted integrand of (\ref{2onee}) is given
instead as $\mathrm{M}(\boldsymbol{\upsilon })=M(\boldsymbol{\upsilon }%
)\otimes \mathrm{Q}^{\otimes }$ to obtain%
\begin{equation*}
\left[ \boldsymbol{\imath }_{0}^{\infty }(M\otimes \mathrm{Q}^{\otimes
})\chi \right] \left( \vartheta \right) =[\boldsymbol{\epsilon }(T)\chi
](\vartheta ).
\end{equation*}%
Here, the operator $\mathrm{Q}^{\otimes}$, given on chains $\vartheta\in%
\mathcal{X}$ as $\mathrm{Q}^\otimes(\vartheta)=\otimes_{x\in\vartheta}%
\mathrm{Q}(x)$ in the space $\mathfrak{k}^\otimes(\vartheta)$, is defined as
a mapping $\mathit{F}_{+}\rightarrow F_-$ such that there exists a $p\in%
\mathfrak{p}$ with
\begin{equation*}
\|\mathrm{Q}^\otimes\|_p=\sup_{\chi\in F_\star(p)}\left\{\frac{\|\mathrm{Q}%
^\otimes\chi\|(p^{-1})}{\|\chi\|(p)}\right\}<\infty.
\end{equation*}
In fact the vacuum projector corresponds to the case $\mathrm{Q}=\mathrm{O}$
such that $\mathrm{P}_\emptyset=\mathrm{O}^\otimes$, and then of course $M(%
\boldsymbol{\upsilon})=T(\boldsymbol{\upsilon})$. The kernel $M(\boldsymbol{%
\upsilon })$ is Maassen-Meyer kernel integrand generalized to the trivially $%
\mathrm{Q}$-adapted case. Indeed, it is uniquely $\mathrm{Q}$-related to $T(%
\boldsymbol{\vartheta })$, it is given by the transformation
\begin{equation*}
M%
\begin{pmatrix}
\upsilon_{+}^{-}, & \upsilon_{\circ}^{-} \\
\upsilon_{+}^{\circ}, & \upsilon_{\circ}^{\circ}%
\end{pmatrix}
=\sum_{\vartheta\subseteq\upsilon_{\circ}^{\circ}}T%
\begin{pmatrix}
\upsilon_{+}^{-}, & \upsilon_{\circ}^{-} \\
\upsilon_{+}^{\circ}, & \vartheta%
\end{pmatrix}
\otimes(-\mathrm{Q})^{\otimes}(\upsilon_{\circ}^{\circ }\boldsymbol{\setminus%
}\vartheta),
\end{equation*}
that is the Meyer transformation of $\mathrm{T}$ defining the trivially
adapted QS-multiple integrand for the integral representation $\mathrm{T}=%
\boldsymbol{\imath }_{0}^{\infty }(\mathrm{M})$ when $\mathrm{Q}=\mathrm{I}$%
. This simply follows from the definition of the action
\begin{equation*}
[\mathrm{M}(\boldsymbol{\upsilon})\mathring{\chi}(\upsilon^-_\circ\sqcup%
\upsilon^\circ_\circ)](\vartheta^\circ_-) =M(\boldsymbol{\upsilon})\otimes%
\mathrm{Q}^\otimes(\vartheta^\circ_-)\chi
(\upsilon^-_\circ\sqcup\upsilon^\circ_\circ\sqcup\vartheta^\circ_-)
\end{equation*}
on $\mathring{\chi}(\upsilon,\vartheta)=\chi(\upsilon\sqcup\vartheta)$, for
the kernel
\begin{equation*}
T%
\begin{pmatrix}
\vartheta_{+}^{-}, & \vartheta_{\circ}^{-} \\
\vartheta_{+}^{\circ}, & \vartheta_{\circ}^{\circ}%
\end{pmatrix}
=\sum_{\upsilon\subseteq\vartheta_{\circ}^{\circ}}M%
\begin{pmatrix}
\vartheta_{+}^{-}, & \vartheta_{\circ}^{-} \\
\vartheta_{+}^{\circ}, & \upsilon%
\end{pmatrix}
\otimes \mathrm{Q}^{\otimes}(\vartheta_{\circ}^{\circ }\boldsymbol{\setminus}%
\upsilon)
\end{equation*}
that is the M\"obius transformation of $M(\boldsymbol{\upsilon})$ when $%
\mathrm{Q}=\mathrm{I}$, inverting the Meyer transformation.

It was shown in \cite{BelB10}, using the estimate for nonadapted integrals
from {\cite{Be92c},} that if $\mathrm{M}$ is a $q$-contractive ampliation of
$M$, such that $\mathrm{M}=M\otimes\mathrm{Q}^\otimes$ and $\|\mathrm{Q}%
^\otimes\|_q\leq1$, then $\Vert\mathrm{T}\Vert_{p}\leq\Vert
M\Vert_{\infty}^{s}(r)$ for $p\geq r^{-1}+q+s^{-1}$, where $\mathrm{T}=%
\boldsymbol{\imath}^\infty_0(\mathrm{M})\equiv\boldsymbol{\epsilon}(T)$, and
\begin{equation*}
\Vert M\Vert_{t}^{s}(r)=\int_{\mathcal{X}^{t}}\mathrm{d}\upsilon_{+}^{-}%
\left(\int_{\mathcal{X}^{t}}\mathrm{d}\upsilon_{+}^{\circ}\int _{\mathcal{X}%
^{t}}\mathrm{d}\upsilon_{\circ}^{-}\mathrm{ess}\sup
_{\upsilon_{\circ}^{\circ}\in\mathcal{X}^{t}}\{s(\upsilon_{\circ}^{\circ
})\Vert M(\boldsymbol{\upsilon})\Vert\}^{2}r(\upsilon_{+}^{\circ}\sqcup%
\upsilon_{\circ}^{-})\right)^{1/2}.
\end{equation*}
However, using the equivalent representation (\ref{two2}) in the form of the
multiple integral (\ref{2onee}) of $\mathrm{M}(\boldsymbol{\upsilon})=T(%
\boldsymbol{\upsilon })\otimes\mathrm{P}_{\emptyset}$, one may write $\|%
\mathrm{T}\|_p\leq\|\mathrm{M}\|^s_{q,\infty}(r)\leq\|T\|_\frac{1}{s}(r)$,
where
\begin{equation*}
\Vert\mathrm{M}\Vert_{q,t}^{s}(r)=\int_{\mathcal{X}^{t}}\left( \int_{%
\mathcal{X}^{t}}\int_{\mathcal{X}^{t}}\mathrm{ess}\sup_{\upsilon_{\circ
}^{\circ}\in\mathcal{X}^{t}}(s(\upsilon_{\circ}^{\circ})\Vert\mathrm{M}(%
\boldsymbol{\upsilon})\Vert_{q})^{2}r(\upsilon_{+}^{\circ}\sqcup
\upsilon_{\circ}^{-})\mathrm{d}\upsilon_{+}^{\circ}\mathrm{d}\upsilon_{\circ
}^{-}\right) ^{1/2} \mathrm{d}\upsilon^-_+ ,
\end{equation*}
taking into account the fact that $\Vert\mathrm{P}_{\emptyset}\Vert_{q}=1$.
This gives a more precise estimate, with $\Vert\mathrm{T}\Vert _{p}\leq\Vert
T\Vert_{\frac{1}{s}}(r)$ holding for $p\geq
r^{-1}+s^{-1}=\lim_{q_{0}\searrow0}(r^{-1}+q_{0}+s^{-1})$. From this
estimate the previous one simply follows as
\begin{align*}
\Vert\sum_{\upsilon_{\circ}^{\circ}\subseteq\vartheta_{\circ}^{\circ}}M(%
\boldsymbol{\upsilon})\otimes\mathrm{Q}^{\otimes}(\vartheta_{\circ}^{\circ}%
\boldsymbol{\setminus}\upsilon_{\circ}^{\circ })\Vert &
\leq\sum_{\upsilon_{\circ}^{\circ}\subseteq\vartheta_{\circ}^{\circ}}q(%
\vartheta^\circ_\circ\setminus\upsilon^\circ_\circ)\Vert M(\boldsymbol{%
\upsilon})\Vert \\
& \leq (q+s^{-1})(\vartheta_{\circ}^{\circ })\mathrm{ess}\sup_{\upsilon_{%
\circ}^{\circ}\in\mathcal{X}}\{s(\upsilon_{\circ}^{\circ})\Vert M(%
\boldsymbol{\upsilon })\Vert\,
\end{align*}
where $s(\upsilon_{\circ}^{\circ})=\prod_{x\in\upsilon_{\circ}^{\circ}}s(x),%
\;\;\;q(\vartheta^\circ_\circ\setminus\upsilon_{\circ}^{\circ})=
\prod_{x\in\vartheta^\circ_\circ\setminus\upsilon_{\circ}^{\circ}}q(x),$ and
\begin{equation*}
(q+s^{-1})(\vartheta_{\circ}^{\circ})=\sum_{\upsilon_{\circ}^{\circ}%
\subseteq\vartheta_{\circ}^{\circ}}s^{-1}(\upsilon_{\circ}^{\circ
})q(\vartheta^\circ_\circ\setminus\upsilon_{\circ}^{\circ})=\prod_{x\in%
\vartheta_{\circ}^{\circ}}(q(x)+s^{-1}(x)),
\end{equation*}
and consequently $\Vert T\Vert_{p}(r)\leq\Vert M\Vert_{\infty}^{s}(r)$ for $%
p\geq q+1/s$. Hence in particular there follows the existence of the adjoint
operator $\mathrm{T}^{\ast}$ bounded in norm $\Vert\mathrm{T}^{\ast
}\Vert_{p}\leq\Vert T^{\star}\Vert_{q}(r)=\Vert T\Vert_{q}(r)$ as the
representation
\begin{equation}
\boldsymbol{\epsilon}(T)^{\ast}=\boldsymbol{\epsilon}(T^{\star
}),\;\,T^{\star}%
\begin{pmatrix}
\vartheta_{+}^{-}, & \vartheta_{\circ}^{-} \\
\vartheta_{+}^{\circ}, & \vartheta_{\circ}^{\circ}%
\end{pmatrix}
=T%
\begin{pmatrix}
\vartheta_{+}^{-}, & \vartheta_{+}^{\circ} \\
\vartheta_{\circ}^{-}, & \vartheta_{\circ}^{\circ}%
\end{pmatrix}
^{\ast}   \label{two3}
\end{equation}
of the $\star$-adjoint kernel $T^{\star}(\boldsymbol{\vartheta})=T(%
\boldsymbol{\vartheta}^{\prime})^{\ast}$, $(\vartheta^\mu_\nu)^\prime=(%
\vartheta^{-\nu}_{-\mu})$.

\section{The Inductive $\star$-Algebra of Relatively Bounded Kernels}

In the next theorem we prove that the $\star$-map $\boldsymbol{\epsilon }%
:T\mapsto\boldsymbol{\epsilon}(T)$ is an operator representation of the $%
\star$-algebra of triangular kernels $T(\boldsymbol{\vartheta})$ satisfying
the boundedness condition
\begin{equation}
\Vert T\Vert_{\boldsymbol{\alpha}}=\underset{\boldsymbol{\vartheta}%
=(\vartheta_{\nu}^{\mu})}{\mathrm{ess}\sup}\{\Vert T(\boldsymbol{\vartheta }%
)\Vert/\prod_{\mu\leq\nu}\alpha_{\nu}^{\mu}(\vartheta_{\nu}^{\mu})\}<\infty
\label{two4}
\end{equation}
relative to the product of the quadruple $\boldsymbol{\alpha}=(\alpha_{\nu
}^{\mu})_{\nu=\circ,+}^{\mu=-,\circ}$ of positive essentially measurable
product functions $\alpha_{\nu}^{\mu}(\vartheta)=\prod_{x\in\vartheta}%
\alpha_{\nu}^{\mu}(x)$, $\vartheta\in\mathcal{X}$. These are defined by an $%
\mathit{L}^{1}$-integrable function $\alpha_{+}^{-}:\mathbb{X}\rightarrow
\mathbb{R}_{+}$, by $\mathit{L}^{2}$-integrable functions $\alpha_{+}^{\circ
},\alpha_{\circ}^{-}:\mathbb{X}\rightarrow\mathbb{R}_{+}$ with a weight $r>0$%
, $r^{-1}\in\mathfrak{p}_{0}$, and by an $\mathit{L}^{\infty}$-function $%
\alpha_{\circ}^{\circ}:\mathbb{X}\rightarrow\mathbb{R}_{+}$, essentially
bounded by unity relative to some $q\in\mathfrak{p}$:
\begin{align}
\Vert\alpha_{+}^{-}\Vert^{(1)}\;\;&=\int|\alpha^-_+(x)|\mathrm{d}x<\infty,
\notag \\
\Vert\alpha\Vert ^{(2)}(r)&=\Big(\int\alpha(x)^{2}r(x)\mathrm{d}x\Big)%
^{1/2}<\infty,\;\;\;\alpha= \alpha^-_\circ,\alpha^\circ_+  \label{two5} \\
\Vert\alpha_{\circ}^{\circ}\Vert_{q}^{(\infty)}\;\;&=\underset{x}{\mathrm{ess%
}\sup }\frac{|\alpha^\circ_\circ(x)|}{q(x)}\leq1.  \notag
\end{align}
The relative boundedness (\ref{two4}) ensures the projective boundedness of $%
T$ by the inequality $\Vert T\Vert_{q}(r)\leq$
\begin{align}
&\leq\int\mathrm{d}\vartheta_{+}^{-}(\iint\mathrm{ess}\underset{%
\vartheta_{\circ }^{\circ}}{\sup}\{\Vert T\Vert_{\boldsymbol{\alpha}%
}\prod\alpha_{\nu}^{\mu
}(\vartheta_{\nu}^{\mu})/q(\vartheta_{\circ}^{\circ})\}^{2}r(\vartheta
_{+}^{\circ}\sqcup\vartheta_{\circ}^{-})\mathrm{d}\vartheta_{+}^{\circ}%
\mathrm{d}\vartheta_{\circ}^{-})^{1/2}  \notag \\
& =\int\alpha_{+}^{-}(\vartheta)\mathrm{d}\vartheta\Big(\int\alpha_{+}^{%
\circ }(\vartheta)^{2}r(\vartheta)\mathrm{d}\vartheta\int\alpha_{\circ}^{-}(%
\vartheta)^{2}r(\vartheta)\mathrm{d}\vartheta\Big)^{1/2}\mathrm{ess}%
\sup_{\vartheta}\frac{\alpha_{\circ}^{\circ}(\vartheta)}{q(\vartheta)}\Vert
T\Vert_{\boldsymbol{\alpha}}  \label{two6} \\
& =\Vert T\Vert_{\boldsymbol{\alpha}}\exp\left\{
\int(\alpha_{+}^{-}(x)+r(x)(\alpha_{+}^{\circ}(x)^{2}+\alpha_{%
\circ}^{-}(x)^{2})/2)\mathrm{d}x\right\} ,  \notag
\end{align}
where we have taken account of the fact that $\int\alpha(\vartheta )\mathrm{d%
}\vartheta=\exp\left\{\int\alpha(x)\mathrm{d}x\right\}$ for $%
\alpha(\vartheta )=\prod_{x\in\vartheta}\alpha(x)$ and%
\begin{equation*}
\mathrm{ess}\sup_{\vartheta}\{\alpha_{\circ}^{\circ}(\vartheta)/q(\vartheta
)\}=\sup_{n}\mathrm{ess}\sup_{x\in \mathbb{X}^{n}}\prod_{i=1}^{n}\{\alpha_{%
\circ }^{\circ}(x_{i})/q(x_{i})\}=1\text{ if }\alpha_{\circ}^{\circ}\leq q.
\end{equation*}

\begin{lemma}
\label{2L 2} Suppose that the multiple quantum-stochastic integral $\mathrm{T%
}_{t}=\boldsymbol{\imath}_{0}^{t}(\mathrm{M})$ is defined by a kernel
operator-function $\mathrm{M}(\boldsymbol{\upsilon})=\boldsymbol{\epsilon}(M(%
\boldsymbol{\upsilon}))$ with values in the operators of the form \textup{(%
\ref{two2})} for $M(\boldsymbol{\upsilon},\boldsymbol{\varkappa})$ in terms
of
\begin{equation*}
T_{\boldsymbol{\upsilon}}%
\begin{pmatrix}
\varkappa_{+}^{-}, & \varkappa_{\circ}^{-} \\
\varkappa_{+}^{\circ}, & \varkappa_{\circ}^{\circ}%
\end{pmatrix}
=M%
\begin{pmatrix}
\upsilon_{+}^{-}, & \upsilon_{\circ}^{-}, & \varkappa_{+}^{-}, &
\varkappa_{\circ}^{-} \\
\upsilon_{+}^{\circ}, & \upsilon_{\circ}^{\circ}, & \varkappa_{+}^{\circ}, &
\varkappa_{\circ}^{\circ}%
\end{pmatrix}
,\upsilon_{\nu}^{\mu}\in\mathcal{X},
\end{equation*}
for \ fixed $\boldsymbol{\upsilon}$\ and $M(\boldsymbol{\upsilon }):%
\boldsymbol{\varkappa}\mapsto M\left( \boldsymbol{\upsilon}\mathbf{,}%
\boldsymbol{\varkappa}\right) $ is a kernel-valued integrand
\begin{equation*}
M(\boldsymbol{\upsilon}\mathbf{,}\boldsymbol{\varkappa}):\mathfrak{k}%
^{\otimes}(\upsilon_{\circ}^{-}\sqcup\varkappa_{\circ}^{-})\otimes \mathfrak{%
k}^{\otimes}(\upsilon_{\circ}^{\circ}\sqcup\varkappa_{\circ}^{\circ })\otimes%
\mathfrak{h}\rightarrow\mathfrak{k}^{\otimes}(\upsilon_{\circ}^{\circ}\sqcup%
\varkappa_{\circ}^{\circ})\otimes\mathfrak{k}^{\otimes}(\upsilon_{+}^{\circ}%
\sqcup\varkappa_{+}^{\circ})\otimes\mathfrak{h}.
\end{equation*}
Then $\mathrm{T}_{t}=\boldsymbol{\epsilon}(T_{t})$ for the kernel $T_{t}(%
\boldsymbol{\vartheta})=\boldsymbol{\nu}_{0}^{t}(\boldsymbol{\vartheta },M) $
given by the multiple counting integral on the kernel-integrands $M$, that
is
\begin{eqnarray}
\boldsymbol{\imath}_{0}^{t}\circ\boldsymbol{\epsilon}=\boldsymbol{\epsilon}%
\circ\boldsymbol{\nu}_{0}^{t}
\end{eqnarray}
where
\begin{equation}
\boldsymbol{\nu}_{0}^{t}(\boldsymbol{\vartheta},M)=\sum_{\boldsymbol{%
\upsilon }\subseteq\boldsymbol{\vartheta}^{t}}M(\boldsymbol{\upsilon},%
\boldsymbol{\vartheta\setminus\upsilon}),   \label{two7}
\end{equation}
with $\boldsymbol{\vartheta}^{t}=(\mathbb{X}^{t}\cap\vartheta_{\nu}^{\mu})_{%
\nu=\circ,+}^{\mu=-,\circ}$ such that the sum is taken over all possible $%
\upsilon_{\nu}^{\mu}\subseteq \mathbb{X}^{t}\cap\vartheta_{\nu}^{\mu}$ and $%
\mu=-,\circ,\nu=\circ,+$.

If $M(\boldsymbol{\upsilon})$ is relatively bounded in for each $\boldsymbol{%
\upsilon}=\left( \upsilon_{\nu}^{\mu }\right) $ such that
\begin{equation*}
\Vert M(\boldsymbol{\upsilon})\Vert_{\boldsymbol{\gamma}}\leq
c\prod_{\mu,\nu
}\beta_{\nu}^{\mu}(\upsilon_{\nu}^{\mu}),\quad\beta_{\nu}^{\mu}(\upsilon
)=\prod_{x\in\upsilon}\beta_{\nu}^{\mu}(x)
\end{equation*}
for some $c>0$ and a pair of quadruples $\boldsymbol{\beta}%
=(\beta_{\nu}^{\mu})$ , $\beta_{\nu}^{\mu}\geq0$ and $\boldsymbol{\gamma}%
=(\gamma_{\nu}^{\mu})$, $\gamma_{\nu}^{\mu}\geq0$ satisfying the
integrability conditions \textup{(\ref{two5})} for $\boldsymbol{\gamma}$,
then the kernel $T$ is relatively bounded:
\begin{equation*}
\Vert\boldsymbol{\nu}_{0}^{t}(M)\Vert_{\boldsymbol{\alpha}}\leq c
\end{equation*}
if $\alpha_{\nu}^{\mu}(x)\geq
\beta_{\nu}^{\mu}(x)1_{[0,t)}(x)+\gamma_{\nu}^{\mu}(x)$ for all $\mu$, $\nu$%
, where $1_{[0,t)}(x)=1$ if $t(x)<t$ and zero if $t(x)\geq t$. In
particular, the generalized single integral $\boldsymbol{i}_{0}^{t}(\mathbf{D%
})$ of the triangular operator-integrand $\mathbf{D}\left( x\right) =\left[
\mathrm{D}_{\nu}^{\mu}\left( x\right) \right] $, with $\mathrm{D}%
_{\nu}^{\mu}(x)=\boldsymbol{\epsilon}(D_{\nu}^{\mu}(x))$, is a
representation
\begin{equation*}
\boldsymbol{i}_{0}^{t}\circ\boldsymbol{\epsilon}=\boldsymbol{\epsilon}\circ%
\boldsymbol{n}_{0}^{t}
\end{equation*}
of the single counting integral
\begin{equation*}
\boldsymbol{n}_{0}^{t}(\boldsymbol{\vartheta},D)=\sum_{\boldsymbol{x}\in%
\boldsymbol{\vartheta}^{t}}M(\mathbf{x},\boldsymbol{\vartheta\setminus x}%
),\quad M(\mathbf{x}_{\nu}^{\mu},\boldsymbol{\varkappa})=D_{\nu}^{\mu }(x,%
\boldsymbol{\varkappa}),
\end{equation*}
of the triangular kernel-integrand $D\left( x,\boldsymbol{\varkappa }\right)
=\left[ D_{\nu}^{\mu}\left( x,\boldsymbol{\varkappa}\right) \right] $, where
the sum is taken over all possible $x\in\vartheta_{\nu}^{\mu}\cap \mathbb{X}%
^{t}$ for $\mu=-,\circ$ and $\nu=\circ,+$, and $\boldsymbol{x}=\boldsymbol{%
\upsilon}_{\nu}^{\mu}(x)$ is one of the atomic tables
\begin{equation}
\boldsymbol{x}_{+}^{-}=%
\begin{pmatrix}
x & \emptyset \\
\emptyset & \emptyset%
\end{pmatrix}
,\;\boldsymbol{x}_{+}^{\circ}=%
\begin{pmatrix}
\emptyset & \emptyset \\
x & \emptyset%
\end{pmatrix}
,\;\boldsymbol{x}_{\circ}^{-}=%
\begin{pmatrix}
\emptyset & x \\
\emptyset & \emptyset%
\end{pmatrix}
,\;\boldsymbol{x}_{\circ}^{\circ}=%
\begin{pmatrix}
\emptyset & \emptyset \\
\emptyset & x%
\end{pmatrix}
,   \label{2oneg}
\end{equation}
with indices $\mu(x)=\kappa$, $\nu(x)=\lambda$ defined almost everywhere by
the condition $x\in\upsilon_{\lambda}^{\kappa}$.
\end{lemma}

\begin{proof}
If $M(\boldsymbol{\upsilon},\boldsymbol{\varkappa})$ is an operator-valued
integrand-kernel that is bounded relative to the pair $(\boldsymbol{\beta}%
\mathbf{,}\boldsymbol{\gamma})$ such that $\Vert M\Vert_{\boldsymbol{\beta}%
\mathbf{,}\boldsymbol{\gamma}}\leq c$, then the relatively bounded operator $%
\mathrm{T}_{t}=\boldsymbol{\epsilon}(T_{t})$ is well-defined for $T_{t}=%
\boldsymbol{\nu}_{0}^{t}(M)$, since
\begin{align*}
\Vert T_{t}(\boldsymbol{\vartheta})\Vert &
\leq\sum_{\upsilon_{+}^{-}\subseteq\vartheta_{+}^{-}}^{t(%
\upsilon_{+}^{-})<t}\sum_{\upsilon
_{+}^{\circ}\subseteq\vartheta_{+}^{\circ}}^{t(\upsilon_{+}^{\circ})<t}%
\sum_{\upsilon_{\circ}^{-}\subseteq\vartheta_{\circ}^{-}}^{t(\upsilon_{\circ
}^{-})<t}\sum_{\upsilon_{\circ}^{\circ}\subseteq\vartheta_{\circ}^{%
\circ}}^{t(\upsilon_{\circ}^{\circ})<t}\Vert M(\boldsymbol{\upsilon},%
\boldsymbol{\vartheta\setminus\upsilon})\Vert \\
& \leq
c\prod_{\nu=\circ,+}^{\mu=-,\circ}\sum_{\upsilon_{\nu}^{\mu}\subseteq%
\vartheta_{\nu}^{\mu}}^{t(\upsilon_{\nu}^{\mu})<t}\beta_{\nu}^{\mu
}(\upsilon_{\nu}^{\mu})\gamma_{\nu}^{\mu}(\vartheta_{\nu}^{\mu}\boldsymbol{%
\setminus}\upsilon_{\nu}^{\mu})=c\prod_{\nu=\circ,+}^{\mu=-,\circ}\alpha_{%
\nu}^{\mu}(\vartheta_{\nu}^{\mu}),
\end{align*}
where $\alpha_{\nu}^{\mu}(\vartheta)=\prod_{x\in\vartheta}^{t(x)<t}[\beta
_{\nu}^{\mu}(x)+\gamma_{\nu}^{\mu}(x)]\cdot\prod_{x\in\vartheta}^{t(x)\geq
t}\gamma_{\nu}^{\mu}(x)$ for $\beta_{\nu}^{\mu}(\upsilon)=\prod_{x\in%
\upsilon }\beta_{\nu}^{\mu}(x)$ and $\gamma_{\nu}^{\mu}(\varkappa)=\prod_{x%
\in \varkappa}\gamma_{\nu}^{\mu}(x)$. Applying the representation (\ref{two2}%
) to $T_{t}(\boldsymbol{\vartheta})=\boldsymbol{\nu}_{0}^{t}(\boldsymbol{%
\vartheta },M)$ it is easy to obtain the representation of the operator $%
\boldsymbol{\epsilon}(T_{t})$ in the form of the generalized multiple
integral (\ref{2onee}) of $\mathrm{M}(\boldsymbol{\upsilon})=\boldsymbol{%
\epsilon}(M(\boldsymbol{\upsilon}))$. Indeed, $[\mathrm{T}%
_{t}\chi](\vartheta)=$
\begin{align*}
&
=\sum_{\vartheta_{\circ}^{\circ}\sqcup\vartheta_{+}^{\circ}=\vartheta}\iint%
\sum_{\boldsymbol{\upsilon}\subseteq\boldsymbol{\vartheta}^{t}}M(\boldsymbol{%
\upsilon},\boldsymbol{\vartheta\setminus\upsilon})\chi(\vartheta_{\circ}^{%
\circ}\sqcup\vartheta_{\circ}^{-})\mathrm{d}\vartheta_{\circ}^{-}\mathrm{d}%
\vartheta_{+}^{-} \\
& =\sum_{\upsilon_{\circ}^{\circ}\sqcup\upsilon_{+}^{\circ}\subseteq
\vartheta^{t}}\underset{\mathcal{X}^t\times\mathcal{X}^t}{\iint}\mathrm{d}%
\upsilon_{\circ}^{-}\mathrm{d}\upsilon_{+}^{-}\sum_{\varkappa_{\circ}^{%
\circ}\sqcup\varkappa_{+}^{\circ}=\vartheta_{-}^{\circ}}\iint M(\boldsymbol{%
\upsilon},\boldsymbol{\varkappa})\mathring{\chi}(\upsilon
_{\circ}^{\circ}\sqcup\upsilon_{\circ}^{-},\varkappa_{\circ}^{\circ}\sqcup%
\varkappa_{\circ}^{-})\mathrm{d}\varkappa_{\circ}^{-}\mathrm{d}%
\varkappa_{+}^{-},
\end{align*}
where $\vartheta_{-}^{\circ}=\vartheta\boldsymbol{\setminus}%
(\upsilon_{\circ}^{\circ}\sqcup\upsilon_{+}^{\circ}),\,\mathring{\chi }%
(\upsilon,\varkappa)=\chi(\varkappa\sqcup\upsilon)$. Consequently, $\mathrm{T%
}_{t}=\boldsymbol{\imath}_{0}^{t}(\mathrm{M})$, where
\begin{equation*}
\lbrack\mathrm{M}(\boldsymbol{\upsilon})\mathring{\chi}(\upsilon_{\circ
}^{\circ}\sqcup\upsilon_{\circ}^{-})](\vartheta)=\sum_{\varkappa_{\circ
}^{\circ}\sqcup\varkappa_{+}^{\circ}=\vartheta}\iint M(\boldsymbol{\upsilon }%
,\boldsymbol{\varkappa})\mathring{\chi}(\upsilon_{\circ}^{\circ}\sqcup%
\upsilon_{\circ}^{-},\,\varkappa_{\circ}^{\circ}\sqcup\varkappa_{\circ }^{-})%
\mathrm{d}\varkappa_{\circ}^{-}\mathrm{d}\varkappa_{+}^{-},
\end{equation*}
that is, we have proved that $\boldsymbol{\epsilon}\circ\boldsymbol{\nu}%
_{0}^{t}=\boldsymbol{\imath}_{0}^{t}\circ\boldsymbol{\epsilon}$.

In particular, if $M(\boldsymbol{\upsilon},\boldsymbol{\varkappa})=0$ for $%
\sum|\upsilon_{\nu}^{\mu}|\neq1$, then, obviously
\begin{equation*}
\boldsymbol{\nu}_{0}^{t}(\boldsymbol{\vartheta},M)=\boldsymbol{n}_{0}^{t}(%
\boldsymbol{\vartheta},D),\quad\boldsymbol{\imath}_{0}^{t}(\mathrm{M})=%
\boldsymbol{i}_{0}^{t}(\mathbf{D}),
\end{equation*}
where $M_{\nu}^{\mu}(x,\boldsymbol{\varkappa})=M(\mathbf{x}_{\nu}^{\mu },%
\boldsymbol{\varkappa})$ and $\mathrm{M}(\boldsymbol{\upsilon})=0$ for $%
\sum|\upsilon_{\nu}^{\mu}|\neq1$, $\mathrm{D}_{\nu}^{\mu}(x)=\mathrm{M}(%
\mathbf{x}_{\nu}^{\mu})$. This yields the representation $\boldsymbol{%
\epsilon}\circ\boldsymbol{n}_{0}^{t}=\boldsymbol{i}_{0}^{t}\circ\boldsymbol{%
\epsilon}$ for the single generalized non-adapted integral $\boldsymbol{i}%
^t_0(\mathbf{D})=\int_{\mathbb{X}^t}\Lambda(\mathbf{D},\mathrm{d}x)$, $%
\Lambda(\mathbf{D},\bigtriangleup)=\Lambda^\nu_\mu(\mathrm{D}%
^\mu_\nu,\triangle)$ for $\bigtriangleup=\mathbb{X}^{t}$, in the form of the
sum
\begin{equation*}
\sum_{\mu,\nu}\Lambda_{\mu}^{\nu}(\boldsymbol{\epsilon}(D_{\nu}^{\mu
}),\bigtriangleup)=\boldsymbol{\epsilon}\Big(\sum_{\mu,\nu}N_{\mu}^{\nu
}(D_{\nu}^{\mu},\bigtriangleup)\Big),\quad N_{\mu}^{\nu}(\boldsymbol{%
\vartheta },D,\bigtriangleup)=\sum_{x\in\vartheta_{\nu}^{\mu}\cap%
\bigtriangleup }D(x,\boldsymbol{\vartheta\setminus x}_{\nu}^{\mu})
\end{equation*}
of representations of four kernel measures $N_{\nu}^{\mu}(\boldsymbol{%
\vartheta},D_{\nu}^{\mu},\bigtriangleup)$ that define kernel representations
$\boldsymbol{\epsilon}\circ N(\bigtriangleup)=\Lambda (\bigtriangleup)\circ%
\boldsymbol{\epsilon}$ of the canonical measures $\Lambda(\mathbf{D}%
,\bigtriangleup)$ with $\mathrm{D}_{\nu}^{\mu}(x)=\boldsymbol{\epsilon}%
(D_{\nu }^{\mu}(x))$.
\end{proof}

One may now realize the commutative diagram of the quantum stochastic
calculus (Fig. 1.).

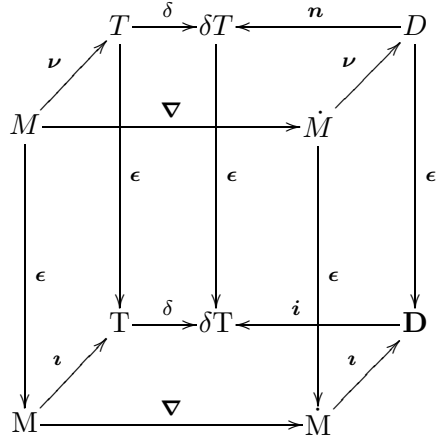
\begin{figure}
\xymatrix@1{ & & & & & T\ar[ddd]^{\boldsymbol{\epsilon}}\ar[r]^\delta& \delta T\ar[ddd]^{\boldsymbol{\epsilon}} & & D\ar[ll]_{\boldsymbol{n}}\ar[ddd]^{\boldsymbol{\epsilon}}\\
& & & & M\ar[ur]^{\boldsymbol{\nu}}\ar[rrr]^{\boldsymbol{\nabla}}\ar[ddd]^{\boldsymbol{\epsilon}}& & & \dot{M}\ar[ur]^{\boldsymbol{\nu}}\ar[ddd]^{\boldsymbol{\epsilon}}\\
 & & & & & & & \\
 & & & & & \mathrm{T}\ar[r]^\delta & \delta\mathrm{T} & & \mathbf{D}\ar[ll]_(0.6){\boldsymbol{i}} \\
 & & & & \mathrm{M}\ar[rrr]^{\boldsymbol{\nabla}}\ar[ur]^{\boldsymbol{\imath}}& & & \dot{\mathrm{M}}\ar[ur]^{\boldsymbol{\imath}}}
\caption{QS calculus commutative diagram. Here $\delta$ has been used to denote the variation $\delta\mathrm{T}=\mathrm{T}_t-\mathrm{T}_0$ of the operator $\mathrm{T}$, and $\nabla$ is the single point split operator \cite{BelB10,Be90f}.}
\end{figure}

\section{Formulation of The Generalized It\^{o} Formula}

In the following theorem, which generalizes the It\^{o} formula to
non-commutative and non-adapted quantum stochastic processes $\mathrm{T}_{t}=%
\boldsymbol{\epsilon}\left( T_{t}\right) $ given by an operator-valued
kernel $T_{t}\left( \boldsymbol{\vartheta}\right) $, we use the following
triangular-matrix notation
\begin{equation*}
\mathbf{T}\left( x\right) =\left[ \mathrm{T}\left( \mathbf{x}_{\nu}^{\mu
}\right) \right] ,\;\;\mathrm{T}\left( \mathbf{x}\right) =\nabla _{\mathbf{x}%
}\left( \mathrm{T}_{t}\right) |_{t=t\left( x\right) }
\end{equation*}
for the \emph{QS point derivative} $\nabla_{\mathbf{x}}\left( \mathrm{T}%
\right) =\boldsymbol{\epsilon}\left( \dot{T}\left( \mathbf{x}\right) \right)
$ given by the QS point split \cite{BelB10,Be90f} of the kernel $\dot{T}%
\left( \mathbf{x},\boldsymbol{\varkappa}\right) =T\left( \boldsymbol{%
\varkappa }\sqcup\boldsymbol{x}\right) $, with $\mathrm{T}_{\nu}^{\mu}\left(
x\right) =\mathrm{T}\left( \mathbf{x}_{\nu}^{\mu}\right) $ equal to zero for
$\mu>\nu$ and $\mathrm{T}_{-}^{-}\left( x\right) =\mathrm{T}_{t\left(
x\right) }=\mathrm{T}_{+}^{+}\left( x\right) $.

We notice that if $T_{t}\left( \boldsymbol{\vartheta}\right) =T_{0}\left(
\boldsymbol{\vartheta}\right) +\boldsymbol{n}_{0}^{t}\left( \boldsymbol{%
\vartheta},D\right) $, corresponding to the single-integral representation $%
\mathrm{T}_{t}-\mathrm{T}_{0}=\boldsymbol{i}_{0}^{t}\left( \mathbf{D}\right)
$\ with $\mathrm{D}_{\nu}^{\mu}\left( x\right) =\boldsymbol{\epsilon}\left(
D_{\nu }^{\mu}\left( x\right) \right) $, then $\dot{T}_{t}(\mathbf{x},%
\boldsymbol{\varkappa})=T_{t}(\boldsymbol{\varkappa}\sqcup\boldsymbol{x})$
may be given by%
\begin{equation*}
\dot{T}_{t}(\mathbf{x},\boldsymbol{\varkappa})=\dot{T}_{t\wedge t_{+}\left(
x\right) }\left( \mathbf{x},\boldsymbol{\varkappa}\right) +\sum _{%
\boldsymbol{z}\in\boldsymbol{\varkappa}}^{t_{+}\left( x\right) \leq t\left(
z\right) <t}D\left( \mathbf{z},\left( \boldsymbol{\varkappa} \setminus
\boldsymbol{z}\right) \sqcup\boldsymbol{x}\right) .
\end{equation*}
Indeed, if $t\in(t\left( x\right) ,t_{+}(x)]$, then $\dot{T}_{t}(\mathbf{x},%
\boldsymbol{\varkappa})=\dot{T}_{t_{+}\left( x\right) }(\mathbf{x},%
\boldsymbol{\varkappa})$ and is therefore not dependent on $t$, where
\begin{equation*}
t_{+}\left( x\right) =\min\left\{ t\left( x^{\prime}\right) >t\left(
x\right) :x^{\prime}\in\sqcup\varkappa_{\nu}^{\mu}\right\} ,
\end{equation*}
and therefore the point-wise right limit%
\begin{equation*}
\dot{T}_{t_{+}\left( x\right) }(\mathbf{x},\boldsymbol{\varkappa}%
):=\lim_{t\searrow t\left( x\right) }\dot{T}_{t}(\mathbf{x},\boldsymbol{%
\varkappa})=\dot{T}_{t\left( x\right) }\left( \mathbf{x},\boldsymbol{%
\varkappa}\right) +D\left( \mathbf{x},\boldsymbol{\varkappa }\right)
\end{equation*}
trivially exists in any uniform topology for each $\mathbf{x}\in\left\{
\mathbf{x}_{\nu}^{\mu}:\mu,\nu=-,\circ,+\right\} $ and $\boldsymbol{%
\varkappa }$; with%
\begin{equation*}
\dot{T}_{t_+\left( x\right) }(\mathbf{x}_{-}^{-},\boldsymbol{\varkappa }%
)=T_{t\left( x\right) }(\boldsymbol{\varkappa})=\dot{T}_{t_+\left( x\right)
}(\mathbf{x}_{+}^{+},\boldsymbol{\varkappa})
\end{equation*}
for $\dot{T}_{t}(\mathbf{x}_{-}^{-},\boldsymbol{\varkappa})=T_{t}(%
\boldsymbol{\varkappa})=\dot{T}_{t}(\mathbf{x}_{+}^{+},\boldsymbol{\varkappa
})$ due to the independentcy of $T\left( \vartheta_{\cdot}^{\cdot}\right) $
on $\vartheta_{-}^{-}$ and $\vartheta_{+}^{+}$.

Now we assume that the QS point split $\nabla_{\mathbf{x}}\left( \mathrm{T}%
_{t}\right) =\boldsymbol{\epsilon }\left( \dot{T}_{t}\left( \mathbf{x}%
\right) \right) $, for $t>t(x)$, has the right limit
\begin{equation*}
\mathrm{T}_+\left( \mathbf{x}\right) :=\underset{t\searrow t\left( x\right)}{%
\lim}\nabla_{\mathbf{x}}\left( \mathrm{T}_{t}\right) \equiv\boldsymbol{%
\epsilon}\left(\dot{T}_{t_+(x)}(\mathbf{x})\right)
\end{equation*}
under the continuity of $\boldsymbol{\epsilon}$ in an appropriate topology
on $\mathrm{T}_t$, as they have trivially the limits
\begin{equation*}
\mathrm{T}_+(\mathbf{x}^-_-)=\boldsymbol{\epsilon}\left( \dot{T}_{t\left(
x\right) }\left( \mathbf{x}_{-}^{-}\right) \right) =\mathrm{T}_{t\left(
x\right) }=\boldsymbol{\epsilon}\left( \dot{T}_{t\left( x\right) }\left(
\mathbf{x}_{+}^{+}\right) \right)=\mathrm{T}_+(\mathbf{x}^+_+)
\end{equation*}
for $\mathbf{x}\in\left\{ \mathbf{x}_{-}^{-},\mathbf{x}_{+}^{+}\right\} $.
The operator-valued triangular matrix function $\mathbf{T}_+(x)=\left[%
\mathrm{T}_+(\mathbf{x}^\mu_\nu)\right]$ of these limits is called the \emph{%
QS germ} of the process $\mathrm{T}$ as an operator-valued function $t\mapsto%
\mathrm{T}_t$. As it is proved in the main theorem, these germ-limits are
given as $\mathrm{T}_+\left( \mathbf{x}\right) =\mathrm{T}\left( \mathbf{x}%
\right) +\mathrm{D}\left( \mathbf{x}\right) $ by the matrix elements $%
\mathrm{D}\left( \mathbf{x}_{\nu}^{\mu}\right) $ of the QS-derivatives $%
\mathbf{D}=\left[ \mathrm{D}\left( \mathbf{x}_{\nu}^{\mu}\right) \right] $
and when the kernels $T_{t}$ are given as the multiple counting integrals $%
T_{t}=\boldsymbol{\nu}_{0}^{t}(M)$ \textup{(\ref{two7})} we obtain the
multiple QS integral representation $\mathrm{T}_{t}=\boldsymbol{\imath}%
_{0}^{t}(\mathrm{M})$ \cite{BelB10} with $\mathrm{M}(\boldsymbol{\upsilon})=%
\boldsymbol{\epsilon}(M(\boldsymbol{\upsilon}))$ defining the matrix
elements $\mathrm{D}_{\nu }^{\mu}\left( x\right) $ as%
\begin{equation*}
\mathrm{D}\left( \mathbf{x}_{\nu}^{\mu}\right) =\boldsymbol{\imath}%
_{0}^{t\left( x\right) }\left( \mathrm{\dot{M}}\left( \mathbf{x}_{\nu
}^{\mu}\right) \right) =\boldsymbol{\epsilon}\left( D\left( \mathbf{x}%
_{\nu}^{\mu}\right) \right)
\end{equation*}
for $x\in \mathbb{X}$ from each atomic table $\boldsymbol{x}_{\nu}^{\mu}\ni x
$ in (\ref{2oneg}).

In the following Main Theorem taken from \cite{Be91},\cite{Be92c} regarding
the It\^{o} product formula in terms of the kernels $T$ we shall adopt the
convention that given a chain $\varkappa =\varkappa _{+}^{-}\sqcup \varkappa
_{\circ }^{-}\sqcup \varkappa _{\circ }^{\circ }\sqcup \varkappa _{+}^{\circ
}$, we have sub-chains $\varkappa ^{\mu }:=\varkappa _{\circ }^{\mu }\sqcup
\varkappa _{+}^{\mu }$, and $\varkappa _{\nu }:=\varkappa _{\nu }^{-}\sqcup
\varkappa _{\nu }^{\circ }$.

\begin{definition}
Let the kernels $X,Y$ be given by the continuous operators $X(\boldsymbol{%
\sigma}):\mathfrak{k}^\otimes(\sigma_\circ)\otimes\mathfrak{h}\rightarrow
\mathfrak{k}^\otimes(\sigma^\circ)\otimes\mathfrak{h}$ and $Y(\boldsymbol{%
\tau}):\mathfrak{k}^\otimes(\tau_\circ)\otimes\mathfrak{h}\rightarrow
\mathfrak{k}^\otimes(\tau^\circ)\otimes\mathfrak{h}$, then the associative
product $X\cdot Y$ \cite{Be90f} is given by
\begin{equation}
[X\cdot Y](\boldsymbol{\vartheta})=\sum^{\sigma^\circ_+\sqcup\upsilon^%
\circ_+=\vartheta^\circ_+}
_{\upsilon^-_\circ\sqcup\tau^-_\circ=\vartheta^-_\circ,}\sum_{\sigma^-_+%
\sqcup\upsilon^-_+\sqcup\tau^-_+=\vartheta^-_+} X(\boldsymbol{\sigma})Y(%
\boldsymbol{\tau}) ,  \label{3onem}
\end{equation}
as a mapping of $\mathfrak{k}^\otimes(\tau_\circ)\otimes\mathfrak{h}$ into $%
\mathfrak{k}^\otimes(\sigma^\circ)\otimes\mathfrak{h}$, where $%
\upsilon^\circ_\circ=\vartheta^\circ_\circ$, $\vartheta^\circ=\sigma^\circ$,
and $\tau_\circ=\vartheta_\circ$, and
\begin{equation*}
\boldsymbol{\sigma}=\left(
\begin{array}{cc}
\sigma^-_+ & \upsilon^- \\
\sigma^\circ_+ & \upsilon^\circ
\end{array}
\right),\quad \boldsymbol{\tau}=\left(
\begin{array}{cc}
\tau^-_+ & \tau_\circ^- \\
\upsilon_+ & \upsilon_\circ
\end{array}
\right),
\end{equation*}
with unit kernel $\hat{I}$ given by operators $\hat{I}(\boldsymbol{\vartheta}%
)=\hat{\mathrm{I}}(\vartheta^\circ_\circ)\otimes\delta_
\emptyset(\vartheta\setminus\vartheta^\circ_\circ)$ such that $\hat{I}\cdot
X=X=X\cdot \hat{I}$.
\end{definition}

\begin{maintheorem}
\label{2T 2} $(\mathrm{i})\quad$ If kernel $T(\boldsymbol{\vartheta})$ is
relatively bounded, then the same is true for the kernel $T^{\star}(%
\boldsymbol{\vartheta }):\Vert T^{\star}\Vert_{\boldsymbol{\gamma}}=\Vert
T\Vert_{\boldsymbol{\gamma }^{\prime}}$, where
\begin{equation*}
\boldsymbol{\gamma}=%
\begin{pmatrix}
\gamma_{+}^{-} & \gamma_{\circ}^{-} \\
\gamma_{+}^{\circ} & \gamma_{\circ}^{\circ}%
\end{pmatrix}
,\;\;\boldsymbol{\gamma}^{\prime}=%
\begin{pmatrix}
\gamma_{+}^{-} & \gamma_{+}^{\circ} \\
\gamma_{\circ}^{-} & \gamma_{\circ}^{\circ}%
\end{pmatrix}
,
\end{equation*}
and the operator $\mathrm{T}^{\ast}=\boldsymbol{\epsilon}(T^{\star})$, as
well as the operator $\mathrm{T}=\boldsymbol{\epsilon}(T)$, is $p$-bounded
by the estimate \textup{(\ref{two6})} for $p\geq q+1/r$. For any such
kernels $T(\boldsymbol{\upsilon})$ and $T^{\star}(\boldsymbol{\upsilon})$,
bounded relative to the quadruples $\boldsymbol{\alpha}=(\alpha_{\nu}^{\mu})$
and $\boldsymbol{\gamma}=(\gamma_{\nu}^{\mu})$ of functions $%
\alpha_{\nu}^{\mu }(x)$ and$\,\gamma_{\nu}^{\mu}(x)$ satisfying \textup{(\ref%
{two5}),} the operator
\begin{equation*}
\boldsymbol{\epsilon}(T)\boldsymbol{\epsilon}(T)^{\ast}=\boldsymbol{\epsilon}%
(T\cdot T^{\star}),\quad\boldsymbol{\epsilon }({\hat{I}})=\mathrm{\hat{I}}%
\text{,}\;\text{where\ }{\hat{I}}=\mathbf{1}_{\mathfrak{h}}\otimes{I}%
^{\otimes }\text{,}
\end{equation*}
is well-defined as a $\star$-representation of kernel product \textup{(\ref%
{3onem})} having the estimate $\Vert T\cdot T^{\star}\Vert_{\boldsymbol{\beta%
}}\leq\Vert T\Vert_{\boldsymbol{\alpha}}\Vert T^{\star}\Vert_{\boldsymbol{%
\gamma}}$ if $\beta_{\nu}^{\mu}\geq (\boldsymbol{\alpha}\mathbf{\cdot}%
\boldsymbol{\gamma})_{\nu}^{\mu}$, where $(\boldsymbol{\alpha}\mathbf{\cdot}%
\boldsymbol{\gamma})_{\nu}^{\mu}(x)=\sum\alpha_{\kappa}^{\mu}(x)\gamma_{%
\nu}^{\kappa}(x)$ is defined by the product of triangular matrices
\begin{equation*}
\left[
\begin{array}{ccc}
1 & \alpha_{\circ}^{-} & \alpha_{+}^{-} \\
0 & \alpha_{\circ}^{\circ} & \alpha_{+}^{\circ} \\
0 & 0 & 1%
\end{array}
\right] \left[
\begin{array}{ccc}
1 & \gamma_{\circ}^{-} & \gamma_{+}^{-} \\
0 & \gamma_{\circ}^{\circ} & \gamma_{+}^{\circ} \\
0 & 0 & 1%
\end{array}
\right] =\left[
\begin{array}{ccc}
1 & \alpha_{\circ}^{-}\gamma_{\circ}^{\circ}+\gamma_{\circ}^{-}, & \gamma
_{+}^{-}+\alpha_{\circ}^{-}\gamma_{+}^{\circ}+\alpha_{+}^{-} \\
0, & \alpha_{\circ}^{\circ}\gamma_{\circ}^{\circ}, & \alpha_{\circ}^{\circ
}\gamma_{+}^{\circ}+\gamma_{+}^{\circ} \\
0, & 0, & 1%
\end{array}
\right] .
\end{equation*}
$(\mathrm{ii})\quad$ Let $\mathrm{T}_{t}=\boldsymbol{\epsilon}(T_{t})$ with $%
\dot {T}_{t}\left( \mathbf{x,}\boldsymbol{\varkappa}\right) =T_{t}\left(
\boldsymbol{x}\sqcup\boldsymbol{\varkappa}\right) $ having the QS right
limit at $t\searrow t\left( x\right) $. Let $\mathbf{T}(x)=\left[ \nabla _{%
\mathbf{x}_{\nu}^{\mu}}\left( \mathrm{T}_{t\left( x\right) }\right) \right] $
and $\mathbf{T}_+(x)=[\mathrm{T}_+(\mathbf{x}_{\nu}^{\mu})]$ denote
triangular matrices at $\mathbf{x}$ with $t=t\left( x\right) $ having the
operator-valued matrix elements
\begin{equation}
\mathrm{T}\left( \mathbf{x}_{\nu}^{\mu}\right) =\boldsymbol{\epsilon}(\dot{T}%
_{t(x)}(\mathbf{x}_{\nu}^{\mu}))\equiv\mathrm{T}_{\nu}^{\mu}\left( x\right)
,\quad\mathrm{T}_+\left( \mathbf{x}_{\nu}^{\mu}\right) =\boldsymbol{\epsilon}%
(\dot{T}_{t_{+}\left( x\right) }(\mathbf{x}_{\nu}^{\mu}))\equiv\mathrm{G}%
_{\nu}^{\mu}\left( x\right)   \label{two8}
\end{equation}
corresponding to the single point split $\dot{T}_{t}\left( \mathbf{x}_{\nu
}^{\mu}\right) $ at $x\in\mathbf{x}_{\nu}^{\mu}$ with $t\left( x\right) \leq
t$. Then the operator-functions $\mathrm{D}_{\nu}^{\mu}(x)=\mathrm{G}%
_{\nu}^{\mu}\left( x\right) -\mathrm{T}_{\nu}^{\mu}(x)$ are
quantum-stochastic derivatives of the function $t\mapsto\mathrm{T}_{t}$
which define the QS differential $\mathrm{dT}_{t}=\mathrm{d}\boldsymbol{i}%
_{0}^{t}(\mathbf{D})$ in the difference form so that $\mathrm{T}_{t}-\mathrm{%
T}_{0}=\boldsymbol{i}_{0}^{t}(\mathbf{T}_+-\mathbf{T})$. Moreover, $\mathrm{T%
}_{t}^{\ast}-\mathrm{T}_{0}^{\ast}=\boldsymbol{i}_{0}^{t}(\mathbf{T}%
_+^{\ddagger}-\mathbf{T}^{\ddagger})$, and we have the generalized
non-adapted It\^{o} formula
\begin{equation}
\mathrm{T}_{t}\mathrm{T}_{t}^{\ast}-\mathrm{T}_{0}\mathrm{T}_{0}^{\ast }=%
\boldsymbol{i}_{0}^{t}(\mathbf{TD}^{\ddagger}+\mathbf{DT}^{\ddagger }+%
\mathbf{DD}^{\ddagger})=\boldsymbol{i}_{0}^{t}(\mathbf{T}_+\mathbf{T}%
_+^{\ddagger }-\mathbf{TT}^{\ddagger}),   \label{two9}
\end{equation}
where $\mathbf{D}\mapsto\mathbf{D}^{\ddagger}$ is the pseudo-Euclidean
conjugation $[\mathrm{D}_{\nu}^{\mu}(x)]^{\ddagger}=[\mathrm{D}_{-\mu}^{-\nu
}(x)]^{\ast}$ of the triangular operators
\begin{equation*}
\mathbf{T}=\left[
\begin{array}{ccc}
\mathrm{T} & \mathrm{T}_{\circ}^{-} & \mathrm{T}_{+}^{-} \\
0 & \mathrm{T}_{\circ}^{\circ} & \mathrm{T}_{+}^{\circ} \\
0 & 0 & \mathrm{T}%
\end{array}
\right] ,\;\;\mathbf{D}=\left[
\begin{array}{ccc}
0 & \mathrm{D}_{\circ}^{-} & \mathrm{D}_{+}^{\circ} \\
0 & \mathrm{D}_{\circ}^{\circ} & \mathrm{D}_{+}^{\circ} \\
0 & 0 & 0%
\end{array}
\right] ,\;\;\mathbf{T}_+=\left[
\begin{array}{ccc}
\mathrm{T} & \mathrm{G}_{\circ}^{-} & \mathrm{G}_{+}^{-} \\
0 & \mathrm{G}_{\circ}^{\circ} & \mathrm{G}_{+}^{\circ} \\
0 & 0 & \mathrm{T}%
\end{array}
\right]\equiv\mathbf{G}
\end{equation*}
with the standard block-matrix multiplication $(\mathbf{XY})_{\nu}^{\mu
}=\Sigma\mathrm{X}_{\kappa}^{\mu}\mathrm{Y}_{\nu}^{\kappa}$.
\end{maintheorem}

\begin{proof}
$(\mathrm{i})\quad$ The adjoint operators $\boldsymbol{\epsilon}(T)$ and $%
\boldsymbol{\epsilon}(T^{\star})$, defining the $\star$-representation (\ref%
{two2}) with respect to the kernels $T$, bounded in the sense of (\ref{two4}%
) and (\ref{two5}), are $p$-bounded for $p\geq q+1/r$ by the estimate $\Vert%
\boldsymbol{\epsilon}(T)\Vert_{p}\leq\Vert T\Vert_{q}(r)$ and inequality (%
\ref{two6}); this leads to the exponential estimate
\begin{equation*}
\Vert\boldsymbol{\epsilon}(T)\Vert_{p}\leq\Vert T\Vert_{\boldsymbol{\alpha}%
}\exp\left\{\Vert\alpha_{+}^{-}\Vert^{(1)}+{\tfrac{1}{2}}\left(\Vert%
\alpha_{+}^{\circ
}\Vert^{(2)}(r)^{2}+\Vert\alpha_{\circ}^{-}\Vert^{(2)}(r)^{2}\right)\right%
\}.
\end{equation*}
The formula for the kernel multiplication $T^{\star}\cdot T$ which is given
above (\ref{3onem}) is $\boldsymbol{\beta}$-bounded for $\boldsymbol{\beta}=%
\boldsymbol{\alpha}\mathbf{\cdot}\boldsymbol{\gamma}$, since $\left\Vert
[T\cdot T^{\star}](\boldsymbol{\vartheta})\right\Vert \leq$
\begin{align*}
& \leq\sum\left\Vert T\left(
\begin{array}{ll}
\vartheta_{+}^{-}\boldsymbol{\setminus}\sigma_{+}^{-}, & \varkappa
_{\circ}^{-}\sqcup\varkappa_{+}^{-} \\
\vartheta_{+}^{\circ}\boldsymbol{\setminus}\varkappa_{+}^{\circ}, &
\vartheta_{\circ}^{\circ}\sqcup\varkappa_{+}^{\circ}%
\end{array}
\right) \right\Vert \cdot\left\Vert T^{\star}\left(
\begin{array}{ll}
\vartheta_{+}^{-}\boldsymbol{\setminus}\tau_{+}^{-}, & \vartheta_{\circ }^{-}%
\boldsymbol{\setminus}\varkappa_{\circ}^{-} \\
\varkappa_{+}^{-}\sqcup\varkappa_{+}^{\circ}, & \vartheta_{\circ}^{\circ
}\sqcup\varkappa_{\circ}^{-}%
\end{array}
\right) \right\Vert \\
& \leq\left\Vert T\right\Vert _{\boldsymbol{\alpha}}\left\Vert T^{\star
}\right\Vert _{\boldsymbol{\gamma}}\sum\boldsymbol{\alpha}^{\otimes}\left(
\begin{array}{ll}
\vartheta_{+}^{-}\boldsymbol{\setminus}\sigma_{+}^{-}, & \varkappa
_{\circ}^{-}\sqcup\varkappa_{+}^{-} \\
\vartheta_{+}^{\circ}\boldsymbol{\setminus}\varkappa_{+}^{\circ}, &
\vartheta_{\circ}^{\circ}\sqcup\varkappa_{+}^{\circ}%
\end{array}
\right) \boldsymbol{\gamma}^{\otimes}\left(
\begin{array}{ll}
\vartheta_{+}^{-}\boldsymbol{\setminus}\tau_{+}^{-}, & \vartheta_{\circ }^{-}%
\boldsymbol{\setminus}\varkappa_{\circ}^{-} \\
\varkappa_{+}^{-}\sqcup\varkappa_{+}^{\circ}, & \vartheta_{\circ}^{\circ
}\sqcup\varkappa_{\circ}^{-}%
\end{array}
\right) \\
& =\left\Vert T\right\Vert _{\boldsymbol{\alpha}}\left\Vert T^{\star
}\right\Vert _{\boldsymbol{\gamma}}(\boldsymbol{\alpha}\mathbf{\cdot }%
\boldsymbol{\gamma})^{\otimes}(\vartheta);\qquad(\boldsymbol{\alpha}\mathbf{%
\cdot }\boldsymbol{\gamma})_{\nu}^{\mu}=\sum_{\mu\leq\kappa\leq\nu}\alpha_{%
\kappa }^{\mu}\gamma_{\nu}^{\kappa},
\end{align*}
where we have employed the multiplication formula $\boldsymbol{\alpha }%
^{\otimes}\cdot\boldsymbol{\gamma}^{\otimes}=(\boldsymbol{\alpha}\cdot%
\boldsymbol{\gamma})^{\otimes}$ for scalar exponential kernels
\begin{equation*}
\boldsymbol{\beta}^{\otimes}(\boldsymbol{\vartheta})=\prod\beta_{\nu}^{\mu
}(\vartheta_{\nu}^{\mu});\;\beta_{\nu}^{\mu}(\vartheta)=\prod_{x\in\vartheta
}\beta_{\nu}^{\mu}(x):\;(\boldsymbol{\alpha}\mathbf{\cdot}\boldsymbol{\gamma
})_{\nu}^{\mu}(x)=\sum\alpha_{\kappa}^{\mu}(x)\gamma_{\nu}^{\kappa}(x).
\end{equation*}
Using the main formula
\begin{equation}
\int\sum_{\upsilon\subseteq\vartheta}f(\upsilon,\vartheta \boldsymbol{%
\setminus}\upsilon)\mathrm{d}\vartheta=\iint f(\upsilon ,\varkappa)\mathrm{d}%
\upsilon\mathrm{d}\varkappa,\quad\forall f\in \mathit{L}^{1}(\mathcal{X}%
\times\mathcal{X}),   \label{2oned}
\end{equation}
of the scalar integration \cite{BelB10,Be90f}, we write the scalar square of
the action (\ref{two2}) in the form $\left\Vert \boldsymbol{\epsilon}%
(T)\chi\right\Vert ^{2}=$
\begin{align*}
& =\int\left\Vert \sum_{\vartheta_{\circ}^{\circ}\sqcup\vartheta_{+}^{\circ
}=\vartheta^\circ}\iint T(\boldsymbol{\vartheta})\chi(\vartheta_{\circ} )%
\mathrm{d}\vartheta_{+}^{-}\mathrm{d}\vartheta_{\circ}^{-}\right\Vert ^{2}%
\mathrm{d}\vartheta^\circ \\
& =\iint\iint\int\sum_{\sigma_{\circ}^{\circ}\sqcup\sigma^-_{\circ
}=\vartheta^\circ}\sum_{\tau_{\circ}^{\circ}\sqcup\tau_{+}^{\circ}=%
\vartheta^\circ}\left\langle T(\boldsymbol{\sigma^{\prime}}%
)\chi(\sigma^{\circ})\big| T(\boldsymbol{\tau})\chi(\tau_{\circ})\right%
\rangle\mathrm{d}\vartheta^\circ\mathrm{d}\sigma_{+}^{-} \mathrm{d}%
\sigma^{\circ}_{+}\mathrm{d}\tau_{+}^{-}\mathrm{d}\tau_{\circ}^{-} \\
& =\iint\iint\int\left\langle\chi(\sigma^{\circ} )\big| T^\star\left(
\boldsymbol{\sigma} \right) T\left( \boldsymbol{\tau} \right)
\chi(\tau_{\circ}^{-}\sqcup\varkappa_{\circ})\right\rangle\mathrm{d}%
\boldsymbol{\varkappa}\mathrm{d}\sigma_{+}^{-}\mathrm{d}\sigma^{\circ}_{+}\,%
\mathrm{d}\tau_{+}^{-}\mathrm{d}\tau_{\circ}^{-} \\
& =\int\Big\langle \chi(\vartheta^\circ)\Big|\sum_{\vartheta_{\circ}^{\circ}%
\sqcup\vartheta_{+}^{\circ}=\vartheta^\circ}\iint(T^\star\cdot T)(%
\boldsymbol{\vartheta})\chi(\vartheta_{\circ})\mathrm{d}\vartheta_{+}^{-}%
\mathrm{d}\vartheta_{\circ}^{-}\Big\rangle \mathrm{d}\vartheta^\circ\quad\;=%
\;\langle\chi|{\epsilon}(T^\star\cdot T)\chi\rangle,
\end{align*}
where $\varkappa_{\circ}^{\circ}=\sigma_{\circ}^{\circ}\cap\tau_{\circ}^{%
\circ}$, $\varkappa_{+}^{\circ}=\sigma_{\circ}^{\circ}\cap\tau_{+}^{\circ}$,
$\varkappa_{\circ}^{-}=\tau_{\circ}^{\circ}\cap\sigma^-_{\circ}$, $%
\varkappa_{+}^{-}=\sigma^-_{\circ}\cap\tau_{+}^{\circ}$,
\begin{equation*}
\boldsymbol{\sigma}=\left(
\begin{array}{cc}
\sigma^-_+ & \varkappa^- \\
\sigma^\circ_+ & \varkappa^\circ
\end{array}
\right),\quad \boldsymbol{\tau}=\left(
\begin{array}{cc}
\tau^-_+ & \tau_\circ^- \\
\varkappa_+ & \varkappa_\circ
\end{array}
\right)
\end{equation*}
and the integral over $\mathrm{d}\vartheta^\circ$ of the double sum%
\begin{equation*}
\sum_{\sigma_{\circ}^{\circ}\sqcup\sigma^-_{\circ}=\vartheta^\circ}\sum
_{\tau_{\circ}^{\circ}\sqcup\tau_{+}^{\circ}=\vartheta^\circ}=\sum_{%
\varkappa
_{\circ}^{\circ}\sqcup\varkappa_{+}^{\circ}\sqcup\varkappa_{\circ}^{-}\sqcup%
\varkappa_{+}^{-}=\vartheta^\circ}
\end{equation*}
is replaced by the quadruple integral over $\mathrm{d}\boldsymbol{\varkappa}=%
\mathrm{d}\varkappa_{\circ}^{\circ}\mathrm{d}\varkappa_{+}^{\circ}\mathrm{d}%
\varkappa_{\circ}^{-}\mathrm{d}\varkappa_{+}^{-}$. In the last line we have
substituted back $\vartheta_\circ^-=\tau^-_\circ\sqcup\varkappa^-_\circ$, $%
\vartheta^\circ_+=\sigma_+^\circ\sqcup\varkappa^\circ_+$, $%
\vartheta^-_+=\sigma^-_+\sqcup\varkappa^-_+\sqcup\tau^-_+$, $%
\varkappa^\circ_\circ=\vartheta^\circ_\circ$, and also $\vartheta^\circ=%
\sigma^\circ$, $\vartheta_\circ=\tau_\circ$. Indeed it follows that $%
\epsilon(T)^\ast\epsilon(T)=\epsilon(T^\star\cdot T)$, from which we also
obtain $\epsilon(\hat{I})=\hat{\mathrm{I}}$ that may also be trivially
verified from (\ref{two2}).$\;\;\;\square$\newline
\linebreak $(\mathrm{ii})\quad$ We shall now consider the stochastic
differential $\mathrm{dT}_{t}$ of the multiple integral $\mathrm{T}_{t}=%
\boldsymbol{\imath}_{0}^{t}(\mathrm{M})$ of the operator function $\mathrm{M}%
(\boldsymbol{\upsilon })=\boldsymbol{\epsilon}(M(\boldsymbol{\upsilon}))$
defined by the quantum-stochastic derivatives
\begin{equation*}
\mathrm{D}_{\nu}^{\mu}(x)=\boldsymbol{\imath}_{0}^{t(x)}\circ \boldsymbol{%
\epsilon}\left( \dot{M}(\mathbf{x}_{\nu}^{\mu})\right) =\boldsymbol{%
\epsilon\circ\nu}_{0}^{t(x)}\left( \dot{M}(\mathbf{x}_{\nu}^{\mu})\right) =%
\boldsymbol{\epsilon}(D_{\nu}^{\mu}(x)),
\end{equation*}
representing the differences of the kernels
\begin{equation*}
D(\mathbf{x}_{\nu}^{\mu},\boldsymbol{\varkappa})=\boldsymbol{\nu}%
_{0}^{t\left( x\right) }\left(\boldsymbol{\varkappa},\dot{M}(\mathbf{x}%
_{\nu}^{\mu})\right)=\dot{T}_{t_{+}(x)}(\mathbf{x}_{\nu}^{\mu },\boldsymbol{%
\varkappa})-\dot{T}_{t(x)}(\mathbf{x}_{\nu}^{\mu},\boldsymbol{\varkappa}).
\end{equation*}
Here $\boldsymbol{\nu}_{0}^{t}\left(\boldsymbol{\varkappa},\dot{M}(\mathbf{x}%
)\right)=\sum_{\boldsymbol{\upsilon}\subseteq\boldsymbol{\varkappa}^{t}}M(%
\boldsymbol{\upsilon}\sqcup\boldsymbol{x},\boldsymbol{\varkappa
\setminus\upsilon})$, $\boldsymbol{x}$ is one of the atomic tables (\ref%
{2oneg}),
\begin{equation*}
\dot{T}_{t(x)}(\mathbf{x},\boldsymbol{\varkappa}) =\sum _{\boldsymbol{%
\upsilon}\subseteq\varkappa^{t(x)}}M(\boldsymbol{\upsilon },(\boldsymbol{%
\varkappa}\sqcup\boldsymbol{x})\boldsymbol{\setminus \upsilon})=T_{t(x)}(%
\boldsymbol{\varkappa}\sqcup\boldsymbol{x}),
\end{equation*}
where $\varkappa^{t(x)}=\varkappa\cap[0,t(x))$, and
\begin{align*}
\dot{T}_{t_{+}(x)}(\mathbf{x},\boldsymbol{\varkappa}) & =\sum _{\boldsymbol{%
\upsilon}\subseteq\varkappa^{t(x)}\sqcup\boldsymbol{x}}M(\boldsymbol{\upsilon%
},(\boldsymbol{\varkappa}\sqcup\boldsymbol{x})\boldsymbol{\setminus\upsilon})
\\
& =T_{t(x)}(\boldsymbol{\varkappa}\sqcup\boldsymbol{x})+\sum _{\boldsymbol{%
\upsilon}\subseteq\boldsymbol{\varkappa}^{t(x)}}M(\boldsymbol{\upsilon}\sqcup%
\boldsymbol{x},\boldsymbol{\varkappa \setminus\upsilon}) \\
& =\dot{T}_{t(x)}(\mathbf{x},\boldsymbol{\varkappa})+\boldsymbol{\nu}%
_{0}^{t(x)}\left(\boldsymbol{\varkappa},\dot{M}(\mathbf{x})\right).
\end{align*}
We note that $T_{t_{+}}(\boldsymbol{\vartheta})=\sum_{\boldsymbol{\upsilon }%
\subseteq\boldsymbol{\vartheta}^{t_{+}}}M(\boldsymbol{\upsilon},\boldsymbol{%
\vartheta\setminus\upsilon})$, where $t_{+}=\min \{t(\boldsymbol{x})>t:%
\boldsymbol{x}\in\boldsymbol{\vartheta}\}$, $\boldsymbol{\vartheta}%
^{t_{+}}=\{\boldsymbol{x}\in\boldsymbol{\vartheta }:t(x)\leq t\}$, so that $%
\dot{T}_{t_{+}(x)}(\mathbf{x,}\boldsymbol{\varkappa })=\dot{T}_{t}(\mathbf{x,%
}\boldsymbol{\varkappa})$ for any $t\in (t(\boldsymbol{x}),t_{+}(\boldsymbol{%
x})]$. Thus the derivatives $\mathrm{D}_{\nu}^{\mu}(x)$, $x\in \mathbb{X}^{t}
$, defining the increment $\mathrm{T}_{t}-\mathrm{T}_{0}=\boldsymbol{i}%
_{0}^{t}(\mathbf{D})$, can be written in the form of the differences
\begin{equation*}
\mathrm{D}_{\nu}^{\mu}(x)=\boldsymbol{\epsilon}\left( \dot{T}_{t_{+}(x)}(%
\mathbf{x}_{\nu}^{\mu})\right) -\boldsymbol{\epsilon}\left( \dot {T}_{t(x)}(%
\mathbf{x}_{\nu}^{\mu})\right)
\end{equation*}
of the operators (\ref{two8}). If we consider $\dot{T}_{t}(\mathbf{x})$ as
one of the four entries $\dot{T}_{t}(x)_{\nu}^{\mu}=\dot{T}_{t}(\mathbf{x}%
_{\nu }^{\mu})$ in the matrix-kernel $\left[ \dot{T}_{t}(x)_{\nu}^{\mu}%
\right] \equiv\mathbf{\dot{T}}_{t}(x)$ with $\dot{T}_{t}(x)_{-}^{-}=T_{t(x)}=%
\dot {T}_{t}(x)_{+}^{+}$, we can define the triangular matrix-functions
\begin{equation*}
\mathbf{T}(x)=\boldsymbol{\epsilon}\left(\mathbf{\dot{T}}_{t\left( x\right)
}(x)\right),\;\;\;\;\mathbf{T}_+(x)=\boldsymbol{\epsilon}\left(\mathbf{\dot{T%
}}_{t_{+}(x)}(x)\right).
\end{equation*}
This allows us to obtain the quantum non-adapted It\^{o} formula in the form
\begin{equation*}
\mathrm{T}_{t}\mathrm{T}_{t}^{\ast}-\mathrm{T}_{0}\mathrm{T}_{0}^{\ast }=%
\boldsymbol{i}_{0}^{t}(\mathbf{TD}^{\ddagger}+\mathbf{DT}^{\ddagger }+%
\mathbf{DD}^{\ddagger}),
\end{equation*}
where $\mathbf{D}(x)=\mathbf{T}_+(x)-\mathbf{T}(x)$ as a consequence of the
fact that the map (\ref{two2}) is a $\star$-homomorphism $\mathrm{T}_{t}%
\mathrm{T}_{t}^{\ast}=\boldsymbol{\epsilon}(T_{t}\cdot T_{t}^{\star})$, and
of the formula (\ref{3onem}) for the product of the operator-valued kernels $%
T_{t}$ and $T_{t}^{\star}$ which can be written in the form
\begin{equation*}
\left[T_{t}\cdot T_{t}^{\star}\right](\boldsymbol{\vartheta}\sqcup%
\boldsymbol{x}_{\nu }^{\mu})=\sum_{\kappa=\mu}^{\nu}\left(\dot{T}%
_{t}(x)_{\kappa}^{\mu}\cdot\dot{T}_{t}^{\star}(x)_{\nu}^{\kappa}\right)(%
\boldsymbol{\vartheta})=\left( \mathbf{\dot {T}}_{t}\cdot\mathbf{\dot{T}}%
_{t}^{\ddag}\right) _{\nu}^{\mu}\left( x,\boldsymbol{\vartheta}\right) .
\end{equation*}
Here the right-hand side is computed as an entry in the product of the
kernel-valued triangular operator $\mathbf{\dot{T}}_{t}(x)=\dot{T}_{t}(%
\mathbf{x}_{\cdot}^{\cdot})$ with $\mathbf{\dot{T}}_{t}^{\ddag}\left(
x\right) =\dot{T}_{t}^{\star}(\mathbf{x}_{\cdot}^{\cdot})$ which defines the
multiplication of the entries in terms of the product of the triangular
operator-valued kernel $\left[ \dot{T}_{t}(\mathbf{x}_{\nu}^{\mu})\right] $
and $\left[\dot{T}_t^\star(\mathbf{x}_{\nu}^{\mu})\right]^{}=\left[\dot{T}_t(%
\mathbf{x}_{-\mu }^{-\nu})^{\ast}\right]$ with $\dot{T}_{t}(\mathbf{x}%
_{-}^{-},\boldsymbol{\vartheta})^{\ast}=T_{t}(\boldsymbol{\vartheta})^{\ast
}=\dot{T}_{t}(\mathbf{x}_{+}^{+},\boldsymbol{\vartheta})^{\ast}$ and $\dot{T}%
(\mathbf{x},\boldsymbol{\vartheta})^{\ast}=T(\boldsymbol{\vartheta }\sqcup%
\boldsymbol{x})^{\ast}$. Indeed, from (\ref{3onem}) we obtain%
\begin{align*}
\lbrack T\cdot T^{\star}](\boldsymbol{\vartheta}\sqcup\boldsymbol{x}_{\circ
}^{\circ}) & =[\dot{T}(\mathbf{x}_{\circ}^{\circ})\cdot\dot{T}^{\star }(%
\mathbf{x}_{\circ}^{\circ})](\boldsymbol{\vartheta}), \\
\lbrack T\cdot T^{\star}](\boldsymbol{\vartheta}\sqcup\boldsymbol{x}%
_{+}^{\circ}) & =[\dot{T}(\mathbf{x}_{\circ}^{\circ})\cdot\dot{T}^{\star }(%
\mathbf{x}_{+}^{\circ})+\dot{T}(\mathbf{x}_{+}^{\circ})\cdot T^{\star }](%
\boldsymbol{\vartheta}), \\
\lbrack T\cdot T^{\star}](\boldsymbol{\vartheta}\sqcup\boldsymbol{x}_{\circ
}^{-}) & =[T\cdot\dot{T}^{\star}(\mathbf{x}_{\circ}^{-})+\dot{T}(\mathbf{x}%
_{\circ}^{-})\cdot\dot{T}^{\star}(\mathbf{x}_{\circ}^{\circ })](\boldsymbol{%
\vartheta}), \\
\lbrack T\cdot T^{\star}](\boldsymbol{\vartheta}\sqcup\boldsymbol{x}%
_{+}^{-}) & =[T\cdot\dot{T}^{\star}(\mathbf{x}_{+}^{-})+\dot{T}(\mathbf{x}%
_{\circ}^{-})\cdot\dot{T}^{\star}(\mathbf{x}_{+}^{\circ})+\dot{T}(\mathbf{x}%
_{+}^{-})\cdot T^{\star}](\boldsymbol{\vartheta}),
\end{align*}
which are the matrix elements of the kernel-valued triangular operator%
\begin{equation*}
\Big[  \left[ T\cdot T^{\star}\right] (\boldsymbol{\vartheta}\sqcup%
\boldsymbol{x}_{\nu}^{\mu})\Big]  =\left[ \dot{T}(\mathbf{x}%
_{\lambda}^{\mu})\cdot\dot{T}^{\star }(\mathbf{x}_{\nu}^{\lambda})\right]
\left( \boldsymbol{\vartheta}\right) =\left( \mathbf{\dot{T}}\cdot \mathbf{%
\dot{T}}^{\ddagger}\right) \left( x,\boldsymbol{\vartheta }\right) .
\end{equation*}
This allows us to write the operator-valued triangular matrix%
\begin{equation*}
\left[ \nabla_{\mathbf{x}_{\nu}^{\mu}}\boldsymbol{\epsilon}\left(T\cdot
T^{\star }\right)\right] =\sum_{\kappa=\mu}^{\nu}\boldsymbol{\epsilon}\left[
\dot {T}(x)_{\kappa}^{\mu}\cdot\dot{T}^{\star}(x)_{\nu}^{\kappa}\right] =%
\boldsymbol{\epsilon}\left(\mathbf{\dot{T}}\cdot\mathbf{\dot{T}}%
^{\ddagger}\right)(x)
\end{equation*}
as the block-matrix product
\begin{equation*}
\boldsymbol{\epsilon}\left(\mathbf{\dot{T}}(x)\cdot\mathbf{\dot{T}}^{\ddag
}(x)\right)=\boldsymbol{\epsilon}\left(\mathbf{\dot{T}}\left( x\right)
\right)\boldsymbol{\epsilon}\left(\mathbf{\dot{T}}^{\ddagger}(x)\right)=%
\boldsymbol{\epsilon }\left(\mathbf{\dot{T}}(x)\right)\boldsymbol{\epsilon}%
\left(\mathbf{\dot{T}}(x)\right)^{\ddagger}
\end{equation*}
of $\mathbf{T}=\boldsymbol{\epsilon}\left(\mathbf{\dot{T}}\right)$and $%
\mathbf{T}^{\ddagger}=\boldsymbol{\epsilon}\left(\mathbf{\dot{T}}%
\right)^{\ddagger}$. So we have proved that%
\begin{equation*}
\boldsymbol{\nabla}_{x}\left(\boldsymbol{\epsilon}(T\cdot T^{\star})\right)=%
\mathbf{T}\left( x\right) \mathbf{T}^{\ddagger}\left( x\right) \equiv
\boldsymbol{\nabla}_{x}\left(\boldsymbol{\epsilon}(T)\right)\boldsymbol{%
\nabla}_{x}\left(\boldsymbol{\epsilon}(T^{\star})\right)
\end{equation*}
in terms of the germ-matrices $\mathbf{T}\left( x\right) =\boldsymbol{\nabla
}_{x}\left(\boldsymbol{\epsilon}(T)\right)$, $\mathbf{T}^{\ddagger}\left(
x\right) =\boldsymbol{\nabla}_{x}\left(\boldsymbol{\epsilon}%
(T^{\star})\right)$ having the operator entries
\begin{equation*}
\mathbf{T}\left( x\right) _{\nu}^{\mu}=\boldsymbol{\epsilon}\left(\dot{T}%
\left( \mathbf{x}_{\nu}^{\mu}\right) \right ),\;\mathbf{T}^{\ddagger}\left(
x\right) _{\nu}^{\mu}=\boldsymbol{\epsilon}\left(\dot{T}\left( \mathbf{x}%
_{-\mu}^{-\nu }\right) \right )^{\ast}.
\end{equation*}
Thus we evaluate the germ of $\mathrm{T}_{t}\mathrm{T}_{t}^{\ast}$ at $t=t(x)
$ and $t=t_{+}(x)$ and obtain the difference formula
\begin{equation*}
\boldsymbol{\epsilon}\left(\left(\mathbf{\dot{T}}_{t_{+}(x)}\cdot \mathbf{%
\dot{T}}_{t_{+}(x)}^{\ddag}\right)(x)-\left(\mathbf{\dot{T}}_{t(x)}\cdot%
\mathbf{\dot{T}}_{t(x)}^{\ddag }\right)(x)\right)=\mathbf{T}_+(x)\mathbf{T}%
_+^{\ddagger}(x)-\mathbf{T}(x)\mathbf{T}^{\ddagger}(x),
\end{equation*}
which allows us to write the stochastic derivative of the quantum
non-adapted process $\mathrm{T}_{t}\mathrm{T}_{t}^{\ast}$ in the form
\begin{equation*}
\mathrm{d}\left(\mathrm{T}_{t}\mathrm{T}_{t}^{\ast}\right)=\mathrm{d}%
\boldsymbol{i}_{0}^{t}\left(\mathbf{T}_+\mathbf{T}_+^{\ddagger}-\mathbf{TT}%
^{\ddagger}\right),
\end{equation*}
corresponding to (\ref{two9}). The theorem has been proved.
\end{proof}

%\pagebreak

\section{Weak Form and $\mathrm{Q}$-Adapted Quantum It\^{o} Formula}

\begin{proposition}
Using the non-adapted table of stochastic multiplication,
\begin{align*}
\mathbf{T}_+^{\ddagger}\mathbf{T}_+-\mathbf{T}^{\ddagger}\mathbf{T} = & \;
\mathbf{D}^{\ddagger}\mathbf{T}+\mathbf{T}^{\ddagger}\mathbf{D}+\mathbf{D}%
^{\ddagger}\mathbf{D} \\
= &\left[
\begin{array}{ccc}
0, & \mathrm{T}^{\ast}\mathrm{D}_{\circ}^{-}, & \mathrm{T}^{\ast }\mathrm{D}%
_{+}^{-}+\mathrm{D}_{+}^{-\ast}\mathrm{T} \\
0, & 0, & \mathrm{D}_{\circ}^{-\ast}\mathrm{T} \\
0, & 0, & 0%
\end{array}
\right] +\left[
\begin{array}{ccc}
0, & \mathrm{D}_{+}^{\circ\ast}\mathrm{D}_{\circ}^{\circ}, & \mathrm{D}%
_{+}^{\circ\ast}\mathrm{D}_{+}^{\circ} \\
0, & \mathrm{D}_{\circ}^{\circ\ast}\mathrm{D}_{\circ}^{\circ}, & \mathrm{D}%
_{\circ}^{\circ\ast}\mathrm{D}_{+}^{\circ} \\
0, & 0, & 0%
\end{array}
\right] \\
& +\left[
\begin{array}{ccc}
0, & \mathrm{D}_{+}^{\circ\ast}\mathrm{T}_{\circ}^{\circ}+\mathrm{T}%
_{+}^{\circ\ast}\mathrm{D}_{\circ}^{\circ}, & \mathrm{D}_{+}^{\circ\ast }%
\mathrm{T}_{+}^{\circ}+\mathrm{T}_{+}^{\circ\ast}\mathrm{D}_{+}^{\circ} \\
0, & \mathrm{D}_{\circ}^{\circ\ast}\mathrm{T}_{\circ}^{\circ}+\mathrm{T}%
_{\circ}^{\circ\ast}\mathrm{D}_{\circ}^{\circ}, & \mathrm{D}%
_{\circ}^{\circ\ast}\mathrm{T}_{+}^{\circ}+\mathrm{T}_{\circ}^{\circ \ast}%
\mathrm{D}_{+}^{\circ} \\
0, & 0, & 0%
\end{array}
\right],
\end{align*}
we can write \textup{(\ref{two9})} in a weak form $\left\Vert \mathrm{T}%
_{t}\chi\right\Vert ^{2}-\left\Vert \mathrm{T}_{0}\chi\right\Vert ^{2}=$
\begin{align}
= &\int_{\mathbb{X}^{t}}2\boldsymbol{\Re}\left\langle \mathrm{T}%
_{t(x)}\chi\mid\mathrm{D}_{+}^{-}(x)\chi+\mathrm{D}_{\circ}^{-}(x)\mathring{%
\chi }\left( x\right) \right\rangle \mathrm{d}x +\int_{\mathbb{X}^{t}}
\left\Vert \mathrm{D}_{+}^{\circ}(x)\chi +\mathrm{D}_{\circ}^{\circ}(x)%
\mathring{\chi}\left( x\right) \right\Vert ^{2}\mathrm{d}x  \notag \\
&\;\;\;\;\;\;\;\;\;\;\;\;\;\;\;\;\;\;\;\;\;\;\;\;\;\; \;\;\;\;\;\;+\int_{%
\mathbb{X}^t}2\boldsymbol{\Re}\left\langle \nabla_{x}\mathrm{T}%
_{t(x)}\chi\mid\mathrm{D}_{+}^{\circ}(x)\chi+\mathrm{D}_{\circ}^{\circ}(x)%
\mathring{\chi}\left( x\right) \right\rangle \mathrm{d}x,  \label{two10}
\end{align}
where $\nabla_{x}\mathrm{T}_{t(x)}\chi=\mathrm{T}_{+}^{\circ}(x)\chi +%
\mathrm{T}_{\circ}^{\circ}(x)\mathring{\chi}\left( x\right) $. This formula
is valid for any non-adapted single integral $\mathrm{T}_{t}=\mathrm{T}_{0}+%
\boldsymbol{i}_{0}^{t}(\mathbf{D})$ with square integrable values $\mathrm{T}%
_{t}\chi$ for all $\chi\in\mathit{G}_{+}$ with $\nabla_{x}$ understood as
the Malliavin derivative \cite{Mal78} at the point $x\in \mathbb{X}$
represented in Fock space by $[\nabla_{x}\mathrm{T}_{t(x)}\chi](\vartheta)=[%
\mathrm{T}_{t(x)}\chi](\vartheta\sqcup x)$.
\end{proposition}

\begin{proof}
Indeed, taking into account that
\begin{equation*}
\left\langle \chi\mid\boldsymbol{i}_{0}^{t}(\mathbf{D})\chi\right\rangle
=\int_{\mathbb{X}^{t}}\left[\left\langle \chi\mid\mathrm{D}_{+}^{-}(x)\chi+%
\mathrm{D}_{\circ}^{-}\mathring{\chi}(x)\right\rangle +\left\langle
\mathring{\chi }(x)\mid\mathrm{D}_{+}^{\circ}(x)\chi+\mathrm{D}%
_{\circ}^{\circ}(x)\mathring{\chi}(x)\right\rangle \right]\mathrm{d}x,
\end{equation*}
we readily obtain the weak form of the non-adapted It\^{o} formula if we
substitute $\mathbf{D}^{\ddagger}\mathbf{T}+\mathbf{D}^{\ddagger}\mathbf{D}+%
\mathbf{T}^{\ddagger}\mathbf{D}$ in place of $\mathbf{D}$. But first notice
that we may write
\begin{equation*}
\left\langle \chi\mid\boldsymbol{i}_{0}^{t}(\mathbf{D})\chi\right\rangle=%
\int_{\mathbb{X}^t}\langle \chi\mid\boldsymbol{\nabla}^\ddag_x\mathbf{D}(x)%
\boldsymbol{\nabla}_x\chi\rangle\mathrm{d}x,
\end{equation*}
where $\langle\chi\mid\boldsymbol{\nabla}_x^\ddag=(\chi^\ast,\mathring{\chi}%
^\ast(x),0)$, and $\boldsymbol{\nabla}_x\chi$ is its pseudo-adjoint. Then we
may write the weak form of the non-adapted It\^o formula as
\begin{align*}
\left\Vert \mathrm{T}_{t}\chi\right\Vert ^{2}-\left\Vert \mathrm{T}%
_{0}\chi\right\Vert ^{2}&=\int_{\mathbb{X}^t}\langle \chi\mid\boldsymbol{%
\nabla}^\ddag_x\left( \mathbf{D}^{\ddagger}(x)\mathbf{T}(x)+\mathbf{D}%
^{\ddagger}(x)\mathbf{D}(x)+\mathbf{T}^{\ddagger}(x)\mathbf{D}(x)\right)%
\boldsymbol{\nabla}_x\chi\rangle\mathrm{d}x \\
&=2\Re\int_{\mathbb{X}^t}\langle \chi\mid\boldsymbol{\nabla}^\ddag_x\mathbf{T%
}^{\ddagger}(x)\mathbf{D}(x)\boldsymbol{\nabla}_x\chi\rangle\mathrm{d}x \\
&\;\;\;\;\;\;\;\;\;\;\;\;\;\;\;\;\;\;\;\;\;\;\;+\int_{\mathbb{X}^t}\langle
\chi\mid\boldsymbol{\nabla}^\ddag_x\mathbf{D}^{\ddagger}(x)\mathbf{D}(x)%
\boldsymbol{\nabla}_x\chi\rangle\mathrm{d}x,
\end{align*}
giving us a non-adapted generalization of the It\^o term of the
Hudson-Parthasarathy formula for the adapted integrals in the form
\begin{equation*}
\int_{\mathbb{X}^t}\langle \chi\mid\boldsymbol{\nabla}^\ddag_x\mathbf{D}%
^{\ddagger}(x)\mathbf{D}(x)\boldsymbol{\nabla}_x\chi\rangle\mathrm{d}x
=\int_{\mathbb{X}^{t}}\left\Vert \mathrm{D}_{+}^{\circ}(x)\chi+\mathrm{D}%
_{\circ}^{\circ }(x)\mathring{\chi}\left( x\right) \right\Vert ^{2}\mathrm{d}%
x,
\end{equation*}
and indeed
\begin{align*}
\int_{\mathbb{X}^t}\langle \chi\mid\boldsymbol{\nabla}^\ddag_x\mathbf{T}%
^{\ddagger}(x)\mathbf{D}(x)\boldsymbol{\nabla}_x\chi\rangle\mathrm{d}x&=
\int_{\mathbb{X}^{t}}\left\langle \mathrm{T}_{t(x)}\chi\mid\mathrm{D}%
_{+}^{-}(x)\chi+\mathrm{D}_{\circ}^{-}(x)\mathring{\chi }\left( x\right)
\right\rangle \mathrm{d}x \\
&\;\;\;\;\;\;\;\;+\int_{\mathbb{X}^t}\left\langle \nabla_{x}\mathrm{T}%
_{t(x)}\chi\mid\mathrm{D}_{+}^{\circ}(x)\chi+\mathrm{D}_{\circ}^{\circ}(x)%
\mathring{\chi}\left( x\right) \right\rangle \mathrm{d}x,  \notag
\end{align*}
as long as $\nabla_{x}\mathrm{T}_{t(x)}\chi=\mathrm{T}_{+}^{\circ}(x)\chi+%
\mathrm{T}_{\circ}^{\circ}\left( x\right) \mathring{\chi}\left( x\right) $.
\end{proof}

Note that if $\mathrm{T}_{t}=\boldsymbol{\epsilon}(T_{t})$ is the
representation \textup{(\ref{two2})} of the kernel \textup{(\ref{two1})},
then obviously
\begin{equation*}
\lbrack\boldsymbol{\epsilon}(T_{t})\chi](\vartheta\sqcup x)=\left[%
\boldsymbol{\epsilon}\left(\dot{T}_{t}(\mathbf{x}_{+}^{\circ})\right)\chi +%
\boldsymbol{\epsilon}\left(\dot{T}(\mathbf{x}_{\circ}^{\circ})\right)%
\mathring{\chi }\left( x\right) \right ](\vartheta),
\end{equation*}
and therefore $\nabla_{x}\mathrm{T}_{t(x)}\chi=\mathrm{T}_{+}^{\circ}(x)\chi+%
\mathrm{T}_{\circ}^{\circ}\left( x\right) \mathring{\chi}\left( x\right) $
is satisfied. Also notice that
\begin{equation*}
\int_{\mathbb{X}^t}\langle \chi\mid\boldsymbol{\nabla}^\ddag_x\mathbf{T}%
^{\ddagger}(x)\mathbf{D}(x)\boldsymbol{\nabla}_x\chi\rangle\mathrm{d}x=
\int_{\mathbb{X}^{t}}\langle \chi\mid\mathrm{T}_{t(x)}^\ast\mathrm{dT}%
_{t(x)}\chi\rangle.
\end{equation*}
In the scalar case $\mathfrak{k}_{x}=\mathbb{C}$ for $\mathrm{D}_{+}^{-}=0=%
\mathrm{D}_{\circ}^{\circ}$, $\mathrm{D}_{\circ}^{-}(x)=\mathrm{D}(x)=%
\mathrm{D}_{+}^{\circ}(x)$, and $\mathrm{T}_{\circ }^{\circ}(x)=\mathrm{T}%
_{t(x)}$, $\mathrm{T}_{\circ}^{-}(x)=\mathrm{T}_{+}^{\circ}(x)\equiv[%
\nabla_x,\mathrm{T}_{t(x)}]:=\partial\mathrm{T}\left( x\right) $ we obtain
\begin{equation*}
\left\Vert \mathrm{T}_{t}\chi\right\Vert ^{2}-\left\Vert \mathrm{T}%
_{0}\chi\right\Vert ^{2}=\int_{\mathbb{X}^{t}}2\Re\left\langle \mathrm{T}%
_{t(x)}\chi\mid\mathrm{dT}_{t(x)}\chi\right\rangle+\int_{\mathbb{X}%
^{t}}\left\Vert \mathrm{D}(x)\chi\right\Vert ^{2}\mathrm{d}x,
\end{equation*}
where the first term may be decomposed into the form
\begin{align}
\int_{\mathbb{X}^{t}}2\Re\left\langle \mathrm{T}_{t(x)}\chi\mid\mathrm{dT}%
_{t(x)}\chi\right\rangle&=\int_{\mathbb{X}^{t}}2\Re\langle \chi\mid%
\boldsymbol{\nabla}^\ddag_x\Big(\mathrm{T}^{\ast}_{t(x)}\otimes\mathbf{I}(x)%
\Big) \mathbf{D}(x)\boldsymbol{\nabla}_x\chi\rangle\mathrm{d}x  \notag \\
&\;\;\;\;\;\;\;\;\;\;\;\;\;\;\;\;\;\;\;\;\;\;\;\;+\int_{\mathbb{X}%
^{t}}2\Re\left\langle \partial\mathrm{T}(x)\chi\mid\mathrm{D}%
(x)\chi\right\rangle.  \label{zzz}
\end{align}
The second term of (\ref{zzz}) vanishes in the adapted case, and the first
term is the adapted contribution corresponding to $\mathbf{T}(x)=\mathrm{T}%
_{t(x)}\otimes\mathbf{I}(x)$. This gives the It\^{o} formula for the
normally-ordered non-adapted integral
\begin{equation*}
\mathrm{T}_{t}-\mathrm{T}_{0}=\int_{\mathbb{X}^{t}}(\Lambda_{\circ}^{+}(%
\mathrm{d}x)\mathrm{D}(x)+\mathrm{D}(x)\Lambda_{-}^{\circ}(\mathrm{d}%
x))=\int_{\mathbb{X}^{t}}\mathrm{dT}_{t(x)}
\end{equation*}
with respect to the Wiener stochastic measure $w(\bigtriangleup)$, $%
\bigtriangleup\in\mathfrak{F}_{\mathbb{X}}$, which is represented in $%
\mathcal{G}_{\ast}$ by commuting operators $\widehat{w}(\bigtriangleup)=%
\Lambda_{\circ }^{+}(\bigtriangleup)+\Lambda_{-}^{\circ}(\bigtriangleup)$.
Consider a particular case when the operators $\mathrm{T}_{0},\,\mathbf{D}%
(x) $, and consequently $\mathrm{T}_{t}$ are multiplications by anticipating
functions $T_{0}(w)$,$\,\mathit{D}(x,w)$, and $T_{t}(w)$, of $w$, that is, $%
\mathrm{T}_{0}=T_{0}(\widehat{w})$,$\,\mathrm{D}(x)=\mathit{D}(x,\widehat{w}%
) $, and $\mathrm{T}_{t}=T_{t}(\widehat{w})$. Then the operators $\partial
\mathrm{T}\left( x\right) =[\nabla_{x},\,\mathrm{T}_{t(x)}]=\boldsymbol{%
\epsilon}(\dot{T}_{t(x)}(x))$ are defined by the Malliavin derivative $\dot{T%
}_{t}(x,w)|_{t=t\left( x\right) }$ as the Wiener representation of the point
split $\dot{T}_{t(x)}(x,\vartheta)=T_{t(x)}(x\sqcup\vartheta)$ of
operator-valued kernels in the multiple stochastic integral%
\begin{equation*}
T_{t}(w)=\int T_{t}(\vartheta)w(\mathrm{d}\vartheta)\equiv I_{w}(T_{t}).
\end{equation*}
In this particular case \textup{(\ref{two10})} was also obtained by Nualart
in \textup{\cite{NuaP88}}. Note that in the weak form we can write $\partial%
\mathrm{T}(x)^\ast\mathrm{D}(x){=}\boldsymbol{\nabla}_x^\ddag\boldsymbol{%
\partial}\mathrm{T} (x)^\ddag\mathbf{D}(x)\boldsymbol{\nabla}_x$, where
\begin{equation*}
\boldsymbol{\partial}\mathrm{T}(x):=\left[
\begin{array}{ccc}
0 & \partial\mathrm{T}(x) & 0 \\
0 & 0 & \partial\mathrm{T}(x) \\
0 & 0 & 0
\end{array}
\right].
\end{equation*}

We note that in the $\mathrm{Q}$-adapted case we always have $\mathbf{T}%
\left( x\right) = \mathrm{T}_{t(x)}\otimes\mathbf{Q}\left( x\right) $ with $%
\mathrm{Q}_{\circ}^{\circ}(x)=\mathrm{Q}(x)$, $\mathrm{Q}_{-}^{-}\left(
x\right) =\mathrm{1}=\mathrm{Q}_{+}^{+}\left( x\right) $ and otherwise $%
\mathrm{Q}_{\nu}^{\mu}(x)=0$. Obviously, the product $\mathrm{T}_{t}^{\ast}%
\mathrm{T}_{t}$ for $\mathrm{Q}$-adapted process $\mathrm{T}_{t}$ remains $%
\mathrm{Q}$-adapted iff $\mathrm{Q}$ is orthoprojector $\mathrm{Q}=\mathrm{Q}%
^\ast\mathrm{Q}=\mathrm{Q}^\ast$. Note that if $\mathrm{Q}\neq\mathrm{I}$
then the operator $\partial\mathrm{T}=[\nabla_x,\mathrm{T}_{t(x)}]\neq0$.
However, we may replace this commutator with the $\mathrm{Q}$-commutator
\begin{equation*}
[\nabla_x,\mathrm{T}_{t(x)}]_{\mathrm{Q}}:=\nabla_x\mathrm{T}_{t(x)}-\mathrm{%
Q}(x)\mathrm{T}_{t(x)}\nabla_x,
\end{equation*}
which vanishes when $\mathrm{T}_t$ is $\mathrm{Q}$-adapted.

\begin{corollary}
\label{2C 2} The quantum stochastic process $\mathrm{T}_{t}=\boldsymbol{%
\epsilon}(T_{t})$ is $\mathrm{Q}$-adapted if and only if the kernel process $%
T_{t}$ is $\mathrm{Q}$-adapted in the sense that
\begin{equation*}
T_{t}(\sigma,\varkappa,\tau)=\int T_{t}\left(
\begin{array}{ll}
\vartheta, & \tau \\
\sigma, & \varkappa%
\end{array}
\right) \mathrm{d}\vartheta=T_{t}(\sigma^{t},\varkappa^{t},\tau^{t})\otimes%
\delta_{\emptyset}(\sigma_{ [t})\mathrm{Q}^{\otimes}(\varkappa_{[t})\delta_{%
\emptyset}(\tau_{ [t}),
\end{equation*}
where $\delta_{\emptyset}(\varkappa)=1$ if $\varkappa=\emptyset$, $%
\delta_{\emptyset}(\varkappa)=0$ if $\varkappa\neq\emptyset$, $\varkappa
^{t}=\varkappa\cap \mathbb{X}^{t}$, $\varkappa_{[t}=\{x\in\varkappa:t(x)\geq
t\}$. The quantum-stochastic It\^{o} formula \textup{(\ref{two9})} for such
processes can be written in the strong form
\begin{align*}
\mathrm{T}_{t}^{\ast}\mathrm{T}_{t}-\mathrm{T}_{0}^{\ast}\mathrm{T}_{0} &
=\int_{\mathbb{X}^{t}}(\mathrm{T}_{t(x)}^{\ast}\mathrm{dT}(x)+\mathrm{dT}%
^{\ast}(x)\mathrm{T}_{t(x)}+\mathrm{dT}^{\ast}(x)\mathrm{dT}(x)) \\
& =\boldsymbol{i}_{0}^{t}(\mathbf{T}_+^{\ddagger}\mathbf{T}_+-\mathrm{T}%
^{\ast}\mathrm{T}\otimes\mathbf{Q}^\ddag \mathbf{Q}),
\end{align*}
where $\mathbf{Q}(x)=\left[ \mathrm{Q}_{\nu}^{\mu}\left( x\right) \right] $
is the block-diagonal operator $\mathrm{Q}_{-}^{-}=1=\mathrm{Q}_{+}^{+}$, $%
\mathrm{Q}_{\circ}^{\circ }=\mathrm{Q}$, and
\begin{gather*}
\mathrm{dT}(x)=\Lambda(\mathbf{D},\mathrm{d}x),\,\mathrm{dT}^{\ast
}(x)=\Lambda(\mathbf{D}^{\ddagger},\mathrm{d}x), \\
\mathrm{dT}^{\ast}(x)\mathrm{dT}(x)=\Lambda(\mathbf{D}^{\ddagger}\mathbf{D},%
\mathrm{d}x).
\end{gather*}
This can be written in the weak form \textup{(\ref{two10})}, where $%
\nabla_{x}\mathrm{T}_{t(x)}\chi=[\mathrm{T}_{t(x)}\otimes\mathrm{Q}(x)]%
\mathring{\chi}\left( x\right) $.
\end{corollary}

\begin{remark}
Let the quantum stochastic process $\mathrm{T}_t=\boldsymbol{\epsilon}(T_t)$ be given
by the multiple QS integral of the kernel $\mathrm{M}=\boldsymbol{\epsilon}%
(M)$ such that $T_t=\boldsymbol{\nu}^t_0(M)$, then $\mathrm{T}_t$ is $%
\mathrm{Q}$-adapted if and only if $\mathrm{M}(\boldsymbol{\upsilon})={M}(%
\boldsymbol{\upsilon})\otimes\mathrm{Q}^\otimes$ up to equivalence with
respect to some $\star$-kernel $N$ satisfying null condition $\boldsymbol{\nu%
}^t_0(\boldsymbol{\vartheta},N)=0\;\forall\;\vartheta\in\mathcal{X}$, such
that
\begin{equation}
T_t(\boldsymbol{\vartheta})=T_t(\boldsymbol{\vartheta}^t)\otimes Q(%
\boldsymbol{\vartheta}_{[t})\;\;\Leftrightarrow \;\;M(\boldsymbol{\upsilon},%
\boldsymbol{\varkappa})= M(\boldsymbol{\upsilon})\otimes Q(\boldsymbol{%
\varkappa})+N(\boldsymbol{\upsilon},\boldsymbol{\varkappa}),
\end{equation}
for all $t>0$, and for all chain-tables $\boldsymbol{\upsilon}\sqcup%
\boldsymbol{\varkappa}=\boldsymbol{\vartheta}= \boldsymbol{\vartheta}^t\sqcup%
\boldsymbol{\vartheta}_{[t}$ with $\boldsymbol{\vartheta}^t=\boldsymbol{%
\vartheta}\cap\mathbb{X}^t$, $\boldsymbol{\vartheta}_{[t}=\boldsymbol{%
\vartheta}\cap\mathbb{X}_{[t}$, and $\boldsymbol{\vartheta}_{[t}=\boldsymbol{%
\varkappa}_{[t}\subseteq\boldsymbol{\varkappa}$, and we have identified $M(%
\boldsymbol{\upsilon})\equiv M(\boldsymbol{\upsilon},\boldsymbol{\emptyset})$%
.
\end{remark}

If the integral kernel is of the form $M\otimes Q$ then it follows trivially
that the process $T_t$ is $Q$-adapted. Now consider the $Q$-Meyer transform
of $T_t$ given as
\begin{equation*}
M_t(\boldsymbol{\upsilon})=\sum_{\boldsymbol{\sigma}\sqcup\boldsymbol{\tau}=%
\boldsymbol{\upsilon}}T_t( \boldsymbol{\sigma})\otimes [-Q](\boldsymbol{\tau}%
),
\end{equation*}
where $\boldsymbol{\sigma}=(\sigma^\mu_\nu)$, $\boldsymbol{\tau}%
=(\tau^\mu_\nu)$, $\sigma^\mu_\nu,\tau^\mu_\nu\in{\mathcal{X}}$, and $[-Q](%
\boldsymbol{\varkappa}):=\otimes_{\boldsymbol{x}\in\boldsymbol{\varkappa}}-Q(%
\boldsymbol{x})$, then it follows that $T_t$ is given by the $Q$-M\"obius
transform
\begin{equation*}
T_t(\boldsymbol{\vartheta})=\sum_{\boldsymbol{\upsilon}\subseteq\boldsymbol{%
\vartheta}^t}M_t( \boldsymbol{\upsilon})\otimes Q(\boldsymbol{\vartheta}%
\setminus\boldsymbol{\upsilon}),
\end{equation*}
that is $T_t=\boldsymbol{\nu}^t_0(M_t\otimes Q)$. Now suppose that there is
another kernel $M^{\prime}$ such that we have $T_t=\boldsymbol{\nu}%
^t_0(M^{\prime})$, then by linearity of $\nu^t_0$ we find that $\boldsymbol{%
\nu}^t_0(M_t\otimes Q -M^{\prime})=0$, thus indeed $M^{\prime}=M_t\otimes Q
+N$, where $\boldsymbol{\nu}^t_0(N)=0$. Notice however that we now have the
Maassen-Meyer kernel $M_t$ depending on the terminal time $t$.


\begin{thebibliography}{10}
\bibitem[1]{AccF88} Accardi, L. and Fagnola, F.: \emph{Stochastic Integration%
} % \emph{Lecture Notes in Mathematics}
\textbf{1303}, Springer-Verlag, Berlin Heidelberg New York , 1988.

\bibitem[2]{AccQ89} Accardi, L. and Quaegebeur, J.: {The It\^o Algebra of
Quantum Guassian Fields}, \emph{J.~Funct. Anal.} \textbf{85} (1989) 213--263.

\bibitem[3]{BelB10} Belavkin, V. P. and Brown, M. F.: {$\mathrm{Q}$-adapted
Quantum Stochastic Integrals and Differentials in Fock Scale} (2010), in:
\emph{Proc. 14th Workshop on Non-commutative Harmonic Analysis}, Banach
Center Publications.

\bibitem[4]{Be88} Belavkin, V. P.: {A New Form and $\ast$-Algebraic
Structure of Quantum Stochastic Integrals in Fock Space}, in: \emph{%
Rendiconti del Seminario Matematico e Fisico di Milano LVIII} (1988), Milan.

\bibitem[5]{Be88a} Belavkin, V. P.: {Non-Demolition Measurements, Non-Linear
Filtering, and Dynamic Programming of Quantum Stochastic Processes} (1988),
in: \emph{Proc. of Bellman Continuum Workshop}.

\bibitem[6]{Be88c} Belavkin, V. P.: {Optimal Non-Linear Filtering of Quantum
Signals}, in: \emph{Proc. 9th Conference on Coding Theory and Information
Transmission} (1988) 342--345, University of Odessa.

\bibitem[7]{Be89d} Belavkin, V. P.: {Stochastic Calculus of Quantum
Input-Output Processes and Non-Demolition Filtering}, \emph{J.~Soviet. Math.}
\textbf{56} (1991).

\bibitem[8]{Be90c} Belavkin, V. P.: \emph{A Quantum Stochastic Calculus in
Fock Space of Input and Output Non-Demolition Processes}
%\emph{Lecture Notes in Mathematics}
\textbf{1442}, Springer-Verlag, Berlin Heidelberg New York, 1990.%
%, 99--125.

\bibitem[9]{Be91c} Belavkin, V. P.: {Continuous Non-Demolition Observation,
Quantum Filtering, and Optimal Estimation}, in: \emph{Proc. Quantum Aspects
of Optical Communication} (1990), Paris.

\bibitem[10]{Be90f} Belavkin, V. P.: {Chaotic States and Stochastic
Integration in Quantum Systems}, \emph{Russian Math. Surveys} \textbf{47}
(1992) 53--116.

\bibitem[11]{Be91} Belavkin, V. P.: {A Quantum Non-Adapted It\^o Formula and
Stochastic Analysis in Fock Scale}, \emph{J. Funct. Anal.} \textbf{102}
(1991) 414--447.

\bibitem[12]{Be92} Belavkin, V. P.: {Quantum Stochastic Calculus and Quantum
Non-Linear Filtering}, \emph{J. Multivariate Analysis} \textbf{42} (1992)
171--201.

\bibitem[13]{Be92c} Belavkin, V. P.: {Chaotic States and Stochastic
Integration in Quantum Systems}, \emph{Russian Math. Surveys} \textbf{1}
(1992), 53--116.

\bibitem[14]{BerK88} Berezanski, M. and Kondrat'ev, G.: {Spectral Methods in
Infinite Dimensional Analysis},: (1998) Naukova Dumka, Kiev.

\bibitem[15]{EvaH88} Evans, M. P. and Hudson, R. L.: \emph{Multidimensional
Quantum Diffusions} %\emph{Lecture Notes in Mathematics}
\textbf{1303}, Springer-Verlag, Berlin Heidelberg New York, 1988.

\bibitem[16]{Gui72} Guichardet, A.: \emph{Symmetric Hilbert Spaces and
Related Topics}, Springer-Verlag, Berlin Heidelberg New York, 1972.

\bibitem[17]{Hid80} Hida, T.: \emph{Brownian Motion}, Springer-Verlag,
Berlin Heidelberg New York, 1980.

\bibitem[18]{HudP84} Hudson, R. L. and Parathasarathy, K. R.: {Quantum It\^o
Formula and Stochastic Evolutions}, \emph{Comm. in Math. Phys.} (1984).

\bibitem[19]{LinM88a} Lindsay, J. M. and Maassen, H.: {The Stochastic
Calculus of Bose Noise}, in: \emph{CWI syllabus 32} (1992), Centre for
Mathematics and Computer science, Amsterdam.

\bibitem[20]{Lin90} Lindsay, J. M.: {On Set Convolutions and Integral-Sum
Kernel Operators}, in: \emph{Proc. Int. Conf. on Probability Theory and
Mathematical Statistics} (1990), Vilnius.

\bibitem[21]{Mal78} Malliavin, P.: \emph{Stochastic Calculus of Variations
and Hypoelliptic Operators}, \emph{Proc. Int. Sym. on Stochastic
Differential Equantions} (1978) 195--293, New York.

\bibitem[22]{Mey87} Meyer, P. A.: \emph{\'El\'ements de Probabilit\'es
Quantiques}, %\emph{Lecture Notes in Mathematics}
\textbf{1247}, Springer-Verlag, Berlin Heidelberg New York, 1987. %33--78.

\bibitem[23]{NuaP88} Nualart, D. and Pardoux, E.: {Stochastic Calculus with
Anticipating Integrals}, \emph{Prob. Theory Related Fields} (1988).

\bibitem[24]{ParS86} Parathasarathy, K. R. and Sinha, K. B.: {Stochastic
Integral Representation of Bounded Quantum Martingales in Fock Space}, \emph{%
J.~Funct. Anal.} \textbf{67} (1986) 126--151.

\bibitem[25]{PotS89} Kondratiev, Yu. G., Leukert, P., Potthoff, J., Streit,
L., and Westerkamp, W.: {Generalized Functionals in Gaussian Spaces: The
Characterization Theorem Revisited}, \emph{J.~Funct. Anal.} \textbf{141}
(1996) 301--318.

\bibitem[26]{Sko75} Skorokhod, A. V.: {On A Generalization of The Stochastic
Integral}, \emph{Theory Prob. Appl.} \textbf{20} (1975) 219--233.
\end{thebibliography}
\end{document}